\documentclass[a4paper]{article}
\usepackage{hyperref}
\usepackage{bbding}
\usepackage[T1]{fontenc}
\usepackage[utf8]{inputenc}
\usepackage{aurical}
\usepackage{huncial}
\usepackage{linearb}
\usepackage{textcomp}
\usepackage{amsmath}
\usepackage{amssymb}
\usepackage[amsmath,amsthm,thmmarks]{ntheorem}
\usepackage{amsfonts}
\usepackage{amssymb,epic,eepic}
\usepackage{amsmath}
\usepackage{stmaryrd}
\usepackage{capt-of}
\usepackage{graphicx}
\usepackage[usenames,dvipsnames]{pstricks}
\usepackage{epsfig}
\usepackage{pst-grad}
 \usepackage{pst-plot}
\usepackage{here} 
\usepackage{paralist}
\usepackage{color}

\newtheorem{theorem}{Theorem}[section]
\newtheorem{lemma}[theorem]{Lemma}
\newtheorem{corollary}[theorem]{Corollary}
\newtheorem{definition}[theorem]{Definition}
\newtheorem{remark}[theorem]{Remark}

\begin{document}
\title{Classical-Quantum Arbitrarily Varying Wiretap Channel: Secret Message Transmission under Jamming Attacks}
\author{Holger Boche\\
Lehrstuhl f\"ur Theoretische
Informationstechnik,\\
 Technische Universit\"at M\"unchen,\\
Munich, Germany\\
boche@tum.de \and Minglai Cai\\
Lehrstuhl f\"ur Theoretische
Informationstechnik,\\
 Technische Universit\"at M\"unchen,\\
Munich, Germany\\
minglai.cai@tum.de \and Christian Deppe \\
Fakult\"at f\"ur
 Mathematik,\\ 
Universit\"at Bielefeld,\\
Bielefeld, Germany\\
cdeppe@mathematik.uni-bielefeld.de  \and Janis N\"otzel  \\
 Universitat Aut\`{o}noma de Barcelona,\\
 Barcelona, Spain/\\
Technische Universit\"at Dresden\\
 Dresden, Germany,\\
Janis.Notzel@uab.cat}

\maketitle
\begin{abstract}
We analyze arbitrarily varying classical-quantum wiretap channels. These channels are subject to two attacks at the same time: one passive (eavesdropping), and one active (jamming).
We progress on previous works \cite{Bo/Ca/De} and \cite{Bo/Ca/De2}, by introducing a reduced class of allowed codes that fulfills a more stringent secrecy requirement than earlier definitions.
In addition, we prove that non-symmetrizability of the legal link is sufficient for equality of the deterministic and the common randomness assisted secrecy capacities.
At last, we focus on analytic properties of both secrecy capacities: We completely characterize their discontinuity points, and their super-activation properties.
\end{abstract}

\tableofcontents

\section{Introduction}

In the last few years
the developments in modern communication systems produced 
many results in a short amount of time.
Especially quantum communication systems allow us to exploit new possibilities while at the same time imposing fundamental limitations.

Quantum mechanics differs significantly from classical mechanics,
it has its own laws.  Quantum information theory unifies information theory with quantum mechanic,
generalizing classical information theory to the quantum world. 
The unit of quantum information is   called    the "qubit", the quantum analogue of the classical "bit".
Unlike a bit, which is either "0" or "1", a qubit can be in "superposition",
i.e. both states at the same time, this is a fundamental tool in quantum information and computing.

A quantum
channel is  a communication  channel which can transmit
quantum information. 
 In general,  there are two
ways to represent a quantum  channel with linear algebraic tools, 
either as a sum of several
transformations, or as a single unitary  transformation which
explicitly includes the unobserved environment.

Quantum  channels  can transmit both classical
and quantum information. We consider the capacity of
quantum channels   carrying classical information.
This is equivalent  to  considering
 the capacity of  classical-quantum channels,
where the  classical-quantum channels are quantum  channels whose
sender's inputs are classical variables.  The classical capacity
of quantum channels has been  determined  in
\cite{Ho}, \cite{Ho}, \cite{Sch/Ni},
and  \cite{Sch/Wes}.

Our goal is
  to investigate   in   communication that
	takes place over a quantum channel 
which is, in addition to   the   noise from the environment, subjected 
to the action of a jammer which actively manipulates the states.
The messages ought also to be kept secret from an 
eavesdropper.

A classical-quantum channel with a jammer is called an arbitrarily varying classical-quantum channel,
where the jammer may change his input in every channel use and is not restricted to use a 
repetitive probabilistic strategy. 
In  the model of an arbitrarily varying channel, we consider  a channel  which is not stationary  
and    can change  with every use. We  interpret   this as an   attack of a
jammer.  
It works as follows:    the sender and the receiver have to select their
coding scheme first. After that the jammer makes his choice of the channel state to sabotage the message  transmission.
However, due to the physical properties, we assume
that  the jammer's changes only take  place in a known set. 

The arbitrarily varying channel was first introduced
 in \cite{Bl/Br/Th2}. 
 \cite{Ahl1}  showed a surprising result
which is  known as the Ahlswede Dichotomy:
Either the 
capacity of an arbitrarily varying channel is zero or it equals its shared randomness assisted capacity. 
After the discovery in   \cite{Bl/Br/Th2} it remained as an open question  when the deterministic capacity is positive. 
In \cite{Rei}  a sufficient condition for that has been given, and in \cite{Cs/Na} it is
 proved that this condition is also necessary.
The  Ahlswede Dichotomy demonstrates the importance of shared randomness for communication in a very clear
 form. 

In \cite{Ahl/Bli} the capacity of arbitrarily varying classical-quantum channels is analyzed. 
A lower 
bound of the capacity has been given. An alternative proof of \cite{Ahl/Bli}'s result and a proof of 
the strong converse are given in \cite{Bj/Bo/Ja/No}. In \cite{Ahl/Bj/Bo/No} the Ahlswede Dichotomy for 
the arbitrarily varying classical-quantum channels is established, and a sufficient and necessary 
condition for the zero deterministic capacity is given. In \cite{Bo/No} a simplification of this 
condition for the arbitrarily varying classical-quantum channels is given.

In  the
model of a wiretap channel
 we consider secure communication. This was first introduced   in \cite{Wyn}. We
 interpret the  wiretap channel as a channel with an eavesdropper.
For a discussion 
of the relation of the different security criteria we refer to 
\cite{Bl/La} and \cite{Wi/No/Bo}.

A classical-quantum  channel with an eavesdropper is called a
classical-quantum wiretap channel, its secrecy capacity has been
determined in \cite{De} and \cite{Ca/Wi/Ye}.

 This work is a progress of our
previous papers  \cite{Bo/Ca/De} and \cite{Bo/Ca/De2}, where  we considered  channel robustness against
 jamming and at the same time security against  eavesdropping.
 A classical-quantum channel with  a jammer and at the same time an eavesdropper is called an arbitrarily 
varying classical-quantum wiretap channel. It is defined as a family of pairs of indexed channels 
$\{(W_t,V_t):t=1,\ldots,T\}$ with a common input alphabet and possible different output systems, 
connecting a sender with two receivers, a legal one and a wiretapper, where $t$ is called a channel 
state of the channel pair. The legitimate receiver accesses the output of the first part of the pair, 
i.e., the first channel $W_t$ in the pair, and the wiretapper observes the output of the second part, 
i.e., the second channel $V_t$. A channel state $t$, which varies from symbol to symbol 
in an arbitrary manner, governs both the legal receiver's channel and the wiretap channel. A code for 
the channel conveys information to the legal receiver such that the wiretapper knows nothing about 
the transmitted information. 

In  \cite{Bo/Ca/De} 
we established the Ahlswede Dichotomy for arbitrarily varying classical-quantum
wiretap  channels,  i.e., either the deterministic 
capacity of an arbitrarily varying channel is zero or equals  its randomness assisted capacity.
 Our proof is similar to the proof of the Ahlswede Dichotomy for  arbitrarily varying classical-quantum channels in
\cite{Ahl/Bli}: we build a two-part code word, the first part is used to
create the common randomness for the sender and the legal receiver, the second is used to transmit the message to the legal
receiver.
We also analyzed the secrecy capacity     when   the sender and the
receiver used various resources. 
In \cite{Bo/Ca/De2} 
we determined the secrecy capacities under common
randomness assisted coding of arbitrarily varying classical-quantum
wiretap  channels. We also 
examined when the secrecy
capacity was a continuous function of the system parameters.
  Furthermore, we proved the phenomenon
``super-activation'' for arbitrarily varying classical-quantum
wiretap channels, i.e.,    there were    two channels, both with zero
deterministic secrecy capacity,    such that if they were used together they    allowed perfect
secure transmission    with positive deterministic secrecy capacity.   
Combining the results of these two paper     we get    the formula for
deterministic secrecy capacity of the arbitrarily varying classical-quantum
wiretap   channel. \vspace{0.2cm}

\begin{center}\begin{figure}[H]
\includegraphics[width=0.8\linewidth]{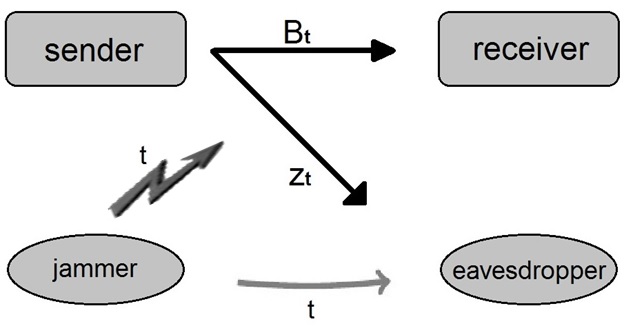}
\caption{Arbitrarily varying classical-quantum wiretap channel}
\end{figure}\end{center}

As aforementioned 
 the lower bound in \cite{Bo/Ca/De} and \cite{Bo/Ca/De2}     is shown    
by building a two-part deterministic code. However 
that code concept still leaves something to be desired
because
we had to reduce the generality of the code concept
when we explicitly allowed a small part of the code word to
be non-secure.
The code word we built was a composition of a public
code word to synchronize the second part and a common randomness assisted code word
to transmit the message.
We only    required    security for the last part.
As we     will   show in 
Corollary \ref{enrtitvwsc}, when
the jammer has access 
to the first part,
it will be rendered
completely useless. Thus
the code concept only works when
the jammer is limited in his action, e.g.
we have to assume that
eavesdropper cannot
send messages towards the jammer.
Nevertheless this code concept with weak criterion 
can be useful when small amount of public messages
is desired, e.g. when  the    receiver    uses it
to estimate the channels.
In this work we 
consider the general code concept
when we construct
a code in such a way that every
part of it is secure.
We show that when the legal channel
is not symmetrizabel the sender can send
a small amount of secure transmissions
which push the secure 
capacity  to the maximally attainable value. 
Thus the entire security is granted.
We call
it the strong code concept.
This completes our analysis of   
arbitrarily     varying classical-quantum wiretap channel.


In \cite{Bo/Ca/De} we analyzed
the secrecy capacities of various coding schemes with resource assistance.
We showed that when the jammer was not allowed to
has access to the resource, it was
very helpful for the secure message transmission
through an arbitrarily varying classical-quantum wiretap channel.
In this work we analyze the case when the shared randomness is not secure
against the jammer.

In \cite{Bo/Ca/De2} we showed that the secrecy
capacity was in general not a continuous 
 function of the system parameters.
In \cite{Bo/Ca/De}
we proved super-activation
for arbitrarily varying classical-quantum wiretap channels.
In this work we establish  complete characterizations for 
  continuity and positivity
of 
the capacity function of arbitrarily varying classical-quantum wiretap channels,
and a complete characterization for
super-activation.

 This
paper is organized as follows:\vspace{0.15cm}

The main definitions  are given in
Section \ref{bdacs2}.

In Section  \ref{CAVQW} we determine a secrecy capacity
formula for a mixed channel model which 
is called the  classic arbitrarily varying quantum wiretap channel.
This formula
is used for
our result in Section  \ref{TSMwSCR}.

In Section  \ref{TSMwSCR} our main result
is presented. In this section we determine the secrecy capacity
for the arbitrarily varying  classical-quantum channels under
strong code concept.

In Section  \ref{cwrioppc} we analyze when the sender and
the legal receiver had the possibility to use shared randomness 
which is not secure against the jammer. We also  determine the secrecy capacity of arbitrarily varying classical-quantum
wiretap  channels  shared randomness 
which is secure against eavesdropping.


As an application of the results of our earlier works, in Section  \ref{SomeApp}
we establish when the secrecy capacity of an arbitrarily varying
classical-quantum wiretap channel is positive and when it  is a  continuous quantity of the
system parameters. Furthermore we show when   ``super-activation'' occurs 
for arbitrarily varying
classical-quantum wiretap channels.

\section{Basic Definitions and Communication Scenarios}
\label{bdacs2}

For  a finite set $\mathbf{A}$ we denote  the
set of probability distributions on $\mathbf{A}$ by $P(\mathbf{A})$.
Let $\rho_1$ and  $\rho_2$ be  Hermitian   operators on a  finite-dimensional
complex Hilbert  space $G$.
We say $\rho_1\geq\rho_2$ and $\rho_2\leq\rho_1$ if $\rho_1-\rho_2$
is positive-semidefinite.
 For a finite-dimensional
complex Hilbert space  $G$, we denote
the (convex) space 
of  density operators on $G$ by
\[\mathcal{S}(G):= \{\rho \in \mathcal{L}(G) :\rho  \text{ is Hermitian, } \rho \geq 0_{G} \text{ , }  \mathrm{tr}(\rho) = 1 \}\text{ ,}\]
where $\mathcal{L}(G)$ is the set  of linear  operators on $G$, and $0_{G}$ is the null
matrix on $G$. Note that any operator in $\mathcal{S}(G)$ is bounded.\vspace{0.15cm}

For  finite sets $\mathbf{A}$  and $\mathbf{B}$,
we define a (discrete) classical channel   $\mathsf{V}$:  $\mathbf{A} \rightarrow P(\mathbf{B})$,
$ \mathbf{A} \ni x \rightarrow \mathsf{V}(x) \in P(\mathbf{B})$
 to be a system characterized by a probability transition matrix $\mathsf{V}(\cdot|\cdot)$.
For $x \in \mathbf{A}$ and $y \in \mathbf{B}$,
$\mathsf{V}(y|x)$  expresses the probability of  the output symbol $y$
when we send the symbol $x$ through the channel. The channel is said to be memoryless
if the probability distribution of the output depends only on the input at
that time and is conditionally independent of previous channel inputs and
outputs.  Further we can extend this definition when we define a 
 classical channel to a map  $\mathsf{V}$:  $P(\mathbf{A}) \rightarrow P(\mathbf{B})$
by denoting $\mathsf{V}(y|p)$ $:=$ $\sum_{x\in\mathbf{A}} p(x)\mathsf{V}(y|x)$.

Let $n \in \mathbb{N}$.
We define
 the $n$-th memoryless extension of
the stochastic matrix $\mathsf{V}$ by $\mathsf{V}^{ n}$, i.e., for $x^n=
(x_1,\dots,x_n)\in A^{ n}$ and  $y^n= (y_1,\dots,y_n)\in B^{ n}$,
$\mathsf{V}^{ n}(y^n|x^n) = \prod_{i=1}^{n} \mathsf{V}(y_i|x_i)$.

 For  finite-dimensional
complex Hilbert spaces  $G$ and  $G'$  a quantum channel $N$:
$\mathcal{S}(G) \rightarrow \mathcal{S}(G')$, $\mathcal{S}(G)  \ni
\rho \rightarrow N(\rho) \in \mathcal{S}(G')$ is represented by a
completely positive trace-preserving map
 which accepts input quantum states in $\mathcal{S}(G)$ and produces output quantum
states in  $\mathcal{S}(G')$.

If the sender wants to transmit a classical message of a finite set $A$ to
the receiver using a quantum channel $N$, his encoding procedure will
include a classical-to-quantum encoder 
to prepare a quantum message state $\rho \in
\mathcal{S}(G)$ suitable as an input for the channel. If the sender's
encoding is restricted to transmit an  indexed finite set of
 quantum states $\{\rho_{x}: x\in \mathbf{A}\}\subset
\mathcal{S}(G)$, then we can consider the choice of the signal
quantum states $\rho_{x}$ as a component of the channel. Thus, we
obtain a channel $\sigma_x := N(\rho_{x})$ with classical inputs $x\in \mathbf{A}$ and quantum outputs,
 which we call a classical-quantum
channel. This is a map $\mathbf{N}$: $\mathbf{A} \rightarrow
\mathcal{S}(G')$, $\mathbf{A} \ni x \rightarrow \mathbf{N}(x) \in
\mathcal{S}(G')$ which is represented by the set of $|\mathbf{A}|$ possible
output quantum states $\left\{\sigma_x = \mathbf{N}(x) :=
N(\rho_{x}): x\in \mathbf{A}\right\}\subset \mathcal{S}(G')$, meaning that
each classical input of $x\in \mathbf{A}$ leads to a distinct quantum output
$\sigma_x \in \mathcal{S}(G')$. In view of this, we have the following
definition.\vspace{0.15cm}

Let $H$ be a finite-dimensional
complex Hilbert space.  A classical-quantum channel is
a linear map $W: P(\mathbf{A})\rightarrow\mathcal{S}(H)$,
$P(\mathbf{A})  \ni P
\rightarrow W(P) \in \mathcal{S}(H)$.
Let $a\in \mathbf{A}$. For a $P_a\in P(\mathbf{A})$, 
defined by $P_a(a')= \begin{cases} 1 &\mbox{if } a'=a\text{ ;}\\
0 &\mbox{if } a'\not=a \end{cases}\text{ ,}$ we 
write $W(a)$ instead of $W(P_a)$.

\begin{remark}    Frequently    a classical-quantum channel is
defined as a map $ \mathbf{A}\rightarrow\mathcal{S}(H)$,
$\mathbf{A}  \ni a
\rightarrow W(a) \in \mathcal{S}(H)$. This is a 
special case when the input is limited on the set $\{P_a :
a\in \mathbf{A}\}$. \end{remark}

For any finite set $\mathbf{A}$, any finite-dimensional
complex Hilbert space $H$, and  $n\in\mathbb{N}$, we define ${\mathbf{A}}^n:= \Bigl\{(a_1,\cdots,a_n): a_i \in \mathbf{A}
\text{ } \forall i \in \{1,\cdots,n\}\Bigr\}$, and $H^{\otimes n}:=
span\Bigl\{v_1\otimes \cdots \otimes v_n: v_i \in H
\text{ } \forall i \in \{1,\cdots,n\}\Bigr\}$. We also write $a^n$ 
for the elements of
${\mathbf{A}}^n$.

Let $n \in \mathbb{N}$.
We define
 the $n$-th memoryless extension of
the stochastic matrix $\mathsf{V}$ by $\mathsf{V}^{ n}$, i.e., for $x^n=
(x_1,\dots,x_n)\in A^{ n}$ and  $y^n= (y_1,\dots,y_n)\in B^{ n}$,
$\mathsf{V}^{ n}(y^n|x^n) = \prod_{i=1}^{n} \mathsf{V}(y_i|x_i)$.

We define the $n$-th extension of
 classical-quantum channel
$W$ as follows.
Associated with $W$ is the channel map on the n-block $W^{\otimes n}$: $P({\mathbf{A}}^n)
\rightarrow \mathcal{S}({H}^{\otimes n})$, such that 
 $W^{\otimes n}(P^n) = W(P_1)
\otimes \cdots \otimes W(P_n)$
if $P^n\in P(\mathbf{A}^{ n})$ can be written as
$(P_1,\cdots,P_n)$. 

Let $\theta$ $:=$ $\{1,\cdots,T\}$ be a finite set.
Let $\Bigl\{W_t:t\in\theta\Bigr\}$ be a set of classical-quantum channels.
For $t^n=(t_1,\cdots,t_n)$, $t_i\in\theta$ we define the n-block $W_{t^n}$ 
such that  for
 $W_{t^n}(P^n) = W_{t_1}(P_1)
\otimes \cdots \otimes W_{t_n}(P_n)$
if $P^n\in P(\mathbf{A}^{ n})$ can be written as
$(P_1,\cdots,P_n)$.

For a discrete random variable $X$  on a finite set $\mathbf{A}$ and a discrete
random variable  $Y$  on  a finite set $\mathbf{B}$  we denote the Shannon entropy
of $X$ by
$H(X)=-\sum_{x \in \mathbf{A}}p(x)\log p(x)$ and the mutual information between $X$
and $Y$ by  
$I(X;Y) = \sum_{x \in \mathbf{A}}\sum_{y \in \mathbf{B}}  p(x,y) \log{ \left(\frac{p(x,y)}{p(x)p(y)} \right) }$.
Here $p(x,y)$ is the joint probability distribution function of $X$ and $Y$, and 
$p(x)$ and $p(y)$ are the marginal probability distribution functions of $X$ and $Y$ respectively,
and ``$\log$''  means logarithm to base $2$.

For a quantum state $\rho\in \mathcal{S}(H)$ we denote the von Neumann
entropy of $\rho$ by \[S(\rho)=- \mathrm{tr}(\rho\log\rho)\text{
,}\] where ``$\log$''  means logarithm to base $2$.

Let $\mathfrak{P}$ and $\mathfrak{Q}$ be quantum systems. We 
denote the Hilbert space of $\mathfrak{P}$ and $\mathfrak{Q}$ by 
$G^\mathfrak{P}$ and $G^\mathfrak{Q}$, respectively. Let $\phi^\mathfrak{PQ}$ be a bipartite
quantum state in $\mathcal{S}(G^\mathfrak{PQ})$. 
We denote the partial
trace over $G^\mathfrak{P}$ by
\[\mathrm{tr}_{\mathfrak{P}}(\phi^\mathfrak{PQ}):= 
\sum_{l} \langle l|_{\mathfrak{P}} \phi^\mathfrak{PQ} |  l \rangle_{\mathfrak{P}}\text{ ,}\]
where $\{ |  l \rangle_{\mathfrak{P}}: l\}$ is an orthonormal basis
of $G^\mathfrak{P}$.
We denote the conditional entropy by
\[S(\mathfrak{P}\mid\mathfrak{Q})_{\rho}:=
S(\phi^\mathfrak{PQ})-S(\phi^\mathfrak{Q})\text{
.}\]
Here $\phi^\mathfrak{Q}=\mathrm{tr}_{\mathfrak{P}}(\phi^\mathfrak{PQ})$.
Let
$\mathtt{V}$: $A \rightarrow
\mathcal{S}(G)$ be a classical-quantum
channel. For $P\in P(A)$
the conditional entropy of the channel for $\mathtt{V}$ with input distribution $P$
is denoted by
 \[S(\mathtt{V}|P) := \sum_{x\in A} P(x)S(\mathtt{V}(x))\text{
.}\]

Let $\Phi := \{\rho_x : x\in \mathbf{A}\}$  be a set of quantum  states
labeled by elements of $\mathbf{A}$. For a probability distribution  $Q$
on $\mathbf{A}$, the    Holevo $\chi$ quantity is defined as
\[\chi(Q;\Phi):= S\left(\sum_{x\in \mathbf{A}} Q(x)\rho_x\right)-
\sum_{x\in \mathbf{A}} Q(x)S\left(\rho_x\right)\text{ .}\]
Note that we can always associate a state 
$\rho^{XY}=\sum_{x}Q(x)|x\rangle\langle x|\otimes \rho_x$ to
$(Q;\Phi)$ such that $\chi(Q;\Phi)=I(X;Y)$ holds for the quantum
mutual information.
For a set $\mathbf{A}$ and a  Hilbert space $G$ let
$\mathbf{V}$: $\mathbf{A} \rightarrow
\mathcal{S}(G)$ be a classical-quantum
channel. For a probability distribution $P$ on $\mathsf{A}$
the    Holevo $\chi$ quantity  of the channel for $\mathbf{V}$ with input distribution $P$
is
defined as
\[\chi(P;\Phi):= S\left(\mathbf{V} (P)\right)-
S\left(\mathbf{V}|P\right)\text{ .}\] 

We denote the identity operator on a space $H$ by $\mathrm{id}_H$
and the symmetric group on $\{1,\cdots,n\}$ by $\mathsf{S}_n$.

For a probability distribution $P$ on a finite set $\mathbf{A}$  and a positive constant $\delta$,
we denote the set of typical sequences by 
\[\mathcal{T}^n_{P,\delta} :=\left\{ a^n \in \mathbf{A}^n: \left\vert \frac{1}{n} N(a'\mid a^n)
- P(a') \right\vert \leq \frac{\delta}{|\mathbf{A}|}\forall a'\in \mathbf{A}\right\}\text{ ,}\]
where $N(a'\mid a^n)$ is the number of occurrences of the symbol $a'$ in the sequence $a^n$.

Let
$G$ be a finite-dimensional
complex Hilbert space.
Let $n \in \mathbb{N}$ and $\alpha > 0$.
We suppose $\rho \in \mathcal{S}(G)$ has
the spectral decomposition
$\rho = \sum_{x} P(x)  |x\rangle\langle x|$,
its
$\alpha$-typical subspace is the subspace spanned
by $\left\{|x^n\rangle, x^n \in {\mathcal{T}}^n_{P, \alpha}\right\}$,
where  $|x^n\rangle:=\otimes_{i=1}^n|x_i\rangle$.  The orthogonal subspace projector onto the
typical subspace is
\[ \Pi_{\rho ,\alpha}=\sum_{x^n \in {\mathcal{T}}^n_{P, \alpha}}|x^n\rangle\langle x^n|\text{ .}\]

Similarly let $\mathbf{A}$ be a finite set, and  $G$ be a finite-dimensional
complex Hilbert space.
Let
$\mathtt{V}$: $\mathbf{A} \rightarrow
\mathcal{S}(G)$ be a classical-quantum
channel. For $a\in\mathbf{A}$  suppose
$\mathtt{V}(a)$ has
the spectral decomposition $\mathtt{V}(a)$
$ =$ $\sum_{j}
V(j|a) |j\rangle\langle j|$
for a stochastic matrix
$V(\cdot|\cdot)$.
 The $\alpha$-conditional typical
subspace of $\mathtt{V}$ for a typical sequence   $a^n$ is the
subspace spanned by
 $\left\{\bigotimes_{a\in\mathbf{A}}|j^{\mathtt{I}_a}\rangle, j^{\mathtt{I}_a} \in \mathcal{T}^{\mathtt{I}_a}_{V(\cdot|a),\delta}\right\}$.
Here $\mathtt{I}_a$ $:=$ $\{i\in\{1,\cdots,n\}: a_i = a\}$ is an indicator set that selects the indices $i$ in the sequence $a^n$
$=$ $(a_1,\cdots,a_n)$ for which the $i$-th
symbol $a_i$ is equal to $a\in\mathbf{A}$.
The subspace is often referred as the $\alpha$-conditional typical
subspace of the state  $\mathtt{V}^{\otimes n}(a^n)$.
  The orthogonal subspace projector onto it is defined as
	\[\Pi_{\mathtt{V}, \alpha}(a^n)=\bigotimes_{a\in\mathbf{A}}
\sum_{j^{\mathtt{I}_a} \in {\cal
T}^{\mathtt{I}_a}_{\mathtt{V}(\cdot \mid a^n),\alpha}}|j^{\mathtt{I}_a} \rangle\langle j^{\mathtt{I}_a}|\text{ .}
\]
The typical subspace 
has following  properties:

For $\sigma \in \mathcal{S}(G^{\otimes n})$ and $\alpha > 0$ 
there are positive constants $\beta(\alpha)$, $\gamma(\alpha)$, 
and $\delta(\alpha)$, depending on $\alpha$ such that
\begin{equation} \label{te1}
\mathrm{tr}\left({\sigma} \Pi_{\sigma ,\alpha}\right) > 1-2^{-n\beta(\alpha)}
\text{ ,}\end{equation}

\begin{equation} \label{te2}
2^{n(S(\sigma)-\delta(\alpha))}\le \mathrm{tr} \left(\Pi_{\sigma ,\alpha}\right)
\le 2^{n(S(\sigma)+\delta(\alpha))}
\text{ ,}\end{equation}

\begin{equation} \label{te3}
2^{-n(S(\sigma)+\gamma(\alpha))} \Pi_{\sigma ,\alpha} \le \Pi_{\sigma ,\alpha}
{\sigma} \Pi_{\sigma ,\alpha}
\le 2^{-n(S(\sigma)-\gamma(\alpha))} \Pi_{\sigma ,\alpha}
\text{ .}\end{equation}

For 
$a^n \in {\mathcal{T}}^n_{P, \alpha}$ 
there are positive constants $\beta(\alpha)'$, $\gamma(\alpha)'$, 
and $\delta(\alpha)'$, depending on $\alpha$ such that

\begin{equation} \label{te4}
\mathrm{tr}\left(\mathtt{V}^{\otimes n}(a^n) \Pi_{\mathtt{V}, \alpha}(a^n)\right)
> 1-2^{-n\beta(\alpha)'}
\text{ ,}\end{equation}

\begin{align} \label{te5}
&2^{-n(S(\mathtt{V}|P)+\gamma(\alpha)')} \Pi_{\mathtt{V}, \alpha}(a^n)
 \le \Pi_{\mathtt{V}, \alpha}(a^n)\mathtt{V}^{\otimes n}(a^n) \Pi_{\mathtt{V}, \alpha}(a^n)\notag\\
 &\le 2^{-n(S(\mathtt{V}|P)-\gamma(\alpha)')} \Pi_{\mathtt{V}, \alpha}(a^n)
\text{ ,}\end{align}
\begin{equation} \label{te6}
2^{n(S(\mathtt{V}|P)-\delta(\alpha)')}\le \mathrm{tr}\left(
\Pi_{\mathtt{V}, \alpha}(a^n) \right)\le 2^{n(S(\mathtt{V}|P)+\delta(\alpha)')}
\text{ .}\end{equation}

For the classical-quantum channel
$\mathtt{V}: P(\mathbf{A}) \rightarrow \mathcal{S}(G)$ and a probability
distribution $P$ on $\mathbf{A}$  we define
 a quantum state $P\mathtt{V}$ $:=$ $\mathtt{V}(P)$ on $\mathcal{S}(G)$.
 For $\alpha > 0$ we define an
orthogonal subspace projector $\Pi_{P\mathtt{V}, \alpha}$
 fulfilling (\ref{te1}), (\ref{te2}), and (\ref{te3}).
Let $x^n\in{\mathcal{T}}^n_{P, \alpha}$.
For $\Pi_{P\mathtt{V}, \alpha}$ there is a positive constant $\beta(\alpha)''$ such that
following inequality holds:
\begin{equation} \label{te7}  \mathrm{tr} \left(  \mathtt{V}^{\otimes n}(x^n) \cdot \Pi_{P\mathtt{V}, \alpha } \right)
 \geq 1-2^{-n\beta(\alpha)''} \text{ .}\end{equation}

We give here a sketch of the proof. For a detailed proof please see \cite{Wil}.

\begin{proof}
(\ref{te1}) holds because 
$\mathrm{tr}\left({\sigma} \Pi_{\sigma ,\alpha}\right)$
$=$ $\mathrm{tr}\left(\Pi_{\sigma ,\alpha}{\sigma} \Pi_{\sigma ,\alpha}\right) $
$=$ $P({\mathcal{T}}^n_{P, \alpha})$.
(\ref{te2}) holds because
$\mathrm{tr} \left(\Pi_{\sigma ,\alpha}\right)$ $=$
$\left\vert {\mathcal{T}}^n_{P, \alpha} \right\vert$.
(\ref{te3}) holds because
$2^{-n(S(\sigma)+\gamma(\alpha))}$ $\le$
$P^n(x^n)$ $\le$ $2^{-n(S(\sigma)-\gamma(\alpha))}$
for $x\in {\mathcal{T}}^n_{P, \alpha}$ and a positive $\gamma(\alpha)$.
(\ref{te4}), (\ref{te5}), and (\ref{te6})
can be obtained in a similar way.
(\ref{te7}) follows from the permutation-invariance of $\Pi_{P\mathtt{V}, \alpha}$.
\end{proof}

\begin{definition}
Let $\mathbf{A}$ and  $\mathbf{B}$ be finite sets, and  $H$ be a finite-dimensional
complex Hilbert spaces. 
Let $\theta$ := $\{1,\dots,T\}$ be a finite set. 
For every $t \in \theta$ let
$\mathsf{W}_{t}$ be a classical  channel $P(\mathbf{A}) \rightarrow P(\mathbf{B})$,
and $W_{t}$   be a classical-quantum channel
$P(\mathbf{A}) \rightarrow \mathcal{S}(H)$.
We call the set of the  classical
channels  $\{\mathsf{W}_t : t \in \theta\}$ an \bf arbitrarily varying
 channel \rm
and the set of the  classical-quantum
channels  $\{W_t : t \in \theta\}$ an \bf arbitrarily varying
classical-quantum  channel \rm
when the channel state $t$ varies from
symbol to symbol in an  arbitrary manner.

\end{definition}

Strictly speaking, the set $\{W_t : t \in \theta\}$ generates the 
arbitrarily varying
classical-quantum  channel $\{W_{t^n} : t^n \in \theta^n\}$.
When the sender inputs a  $P^n \in P({\mathbf{A}}^n)$ into the channel, the receiver
receives the output $W_{t^n}(P^n)
 \in \mathcal{S}(H^{\otimes n})$, where $t^n = (t_1, t_2, \cdots ,
t_n)\in\theta^n$ is the channel state of $W_{t^n}$.

\begin{definition}\label{symmet}
We say that the arbitrarily varying channel
$\{\mathsf{W}_t : t \in \theta\}$ is symmetrizable if
 there exists a
parametrized set of distributions $\{\tau(\cdot\mid a):
 a\in \mathbf{A}\}$ on $\theta$ such that for all $a$, ${a'}\in \mathbf{A}$, and $b \in \mathbf{B}$
\[\sum_{t\in\theta}\tau(t\mid a)\mathsf{W}_{t}(b\mid {a'})=\sum_{t\in\theta}\tau(t\mid {a'})\mathsf{W}_{t}(b\mid a)\text{ .}\]

We say that the arbitrarily varying classical-quantum  channel
$\{W_t : t \in \theta\}$ is symmetrizable if
 there exists a
parametrized set of distributions $\{\tau(\cdot\mid a):
 a\in \mathbf{A}\}$ on $\theta$ such that for all $a$, ${a'}\in \mathbf{A}$,
\[\sum_{t\in\theta}\tau(t\mid a)W_{t}({a'})=\sum_{t\in\theta}\tau(t\mid {a'})W_{t}(a)\text{ .}\]
\end{definition}

\begin{definition}
Let $\mathbf{A}$ and  $\mathbf{B}$ be finite sets, and  $H$ be a finite-dimensional
complex Hilbert spaces.  Let
  $\theta$ $:=$ $\{1,2, \cdots\}$ be an index set. For every $t \in \theta$  let $\mathsf{W}_{t}$   be a classical channel
$P(\mathbf{A}) \rightarrow P(\mathbf{B})$ and ${V}_t$ be a classical-quantum channel $P(\mathbf{A})
\rightarrow \mathcal{S}(H)$. We call the set of the  classical/classical-quantum 
channel pairs  $\{(\mathsf{W}_t,{V}_t): t \in \theta\}$ an \bf classic arbitrarily varying quantum wiretap channel\it, when the
state $t$ varies from symbol to symbol in an  arbitrary
manner, while  the legitimate
receiver accesses the output of the first channel, i.e., $W_t$  in the pair
$(\mathsf{W}_t,{V}_t)$, and the wiretapper observes the output of the second
 channel, i.e.,  ${V}_t$ in the pair $(\mathsf{W}_t,{V}_t)$, respectively.\end{definition}

\begin{definition}
Let $\mathbf{A}$ be a finite set. Let
 $H$ and $H'$ be finite-dimensional
complex Hilbert spaces. Let
  $\theta$ $:=$ $\{1,2, \cdots\}$ be an index set. For every $t \in \theta$  let $W_{t}$   be a classical-quantum channel
$P(\mathbf{A}) \rightarrow \mathcal{S}(H)$ and ${V}_t$ be a classical-quantum channel $P(\mathbf{A})
\rightarrow \mathcal{S}(H')$. We call the set of the  classical-quantum
channel pairs  $\{(W_t,{V}_t): t \in \theta\}$ an \bf arbitrarily
varying classical-quantum wiretap channel\it, when the
state $t$ varies from symbol to symbol in an  arbitrary
manner, while  the legitimate
receiver accesses the output of the first channel, i.e., $W_t$  in the pair
$(W_t,{V}_t)$, and the wiretapper observes the output of the second
 channel, i.e.,  ${V}_t$ in the pair $(W_t,{V}_t)$, respectively.\end{definition}

\begin{definition}
 An $(n, J_n)$   (deterministic)  code $\mathcal{C}$ for the
arbitrarily
varying classical-quantum wiretap channel $\{(W_t,{V}_t): t \in \theta\}$
consists of a stochastic encoder $E$ : $\{
1,\cdots ,J_n\} \rightarrow P({\mathbf{A}}^n)$, 
$j\rightarrow E(\cdot|j)$,
 specified by
a matrix of conditional probabilities $E(\cdot|\cdot)$, and
 a collection of positive-semidefinite operators $\left\{D_j: j\in \{ 1,\cdots ,J_n\}\right\}$
on ${H}^{\otimes n}$,
which is a partition of the identity, i.e. $\sum_{j=1}^{J_n} D_j =
\mathrm{id}_{{H}^{\otimes n}}$. We call these  operators the decoder operators.
\end{definition}
A code is created by the sender and the legal receiver before the
message transmission starts. The sender uses the encoder to encode the
message that he wants to send, while the legal receiver uses the
decoder operators on the channel output to decode the message.

\begin{remark}
 An $(n, J_n)$   deterministic code $\mathcal{C}$ with deterministic encoder
consists of a family of $n$-length strings of symbols $\left(c_j\right)_{j\in
\{1,\cdots ,J_n\}} \in \left({\mathbf{A}}^n\right)^{J_n}$   and
 a collection of positive-semidefinite operators $\left\{D_j: j\in \{ 1,\cdots ,J_n\}\right\}$
on ${H}^{\otimes n}$
which is a partition of the identity.

 The deterministic encoder is a special case of the
 stochastic encoder when we require that for every $j\in
\{1,\cdots ,J_n\}$, there is a sequence $a^n\in {\mathbf{A}}^n$  chosen with probability $1$.
 The standard technique for  message transmission over a  channel and
robust message transmission over an arbitrarily
varying  channel
 is to use the deterministic encoder
(cf.  \cite{Ahl/Bj/Bo/No} and \cite{Bo/No}).
 However, we use  the stochastic encoder, since it is  a tool
 for  secure message transmission over wiretap channels (cf. \cite{Bl/Ca} and \cite{Ahl/Bli}).
\label{detvsran}
\end{remark}

\begin{definition}
A non-negative number $R$ is an achievable   secrecy
rate for the arbitrarily varying classical-quantum wiretap channel
$\{(W_t,{V}_t): t \in \theta\}$ if for every $\epsilon>0$, $\delta>0$,
$\zeta>0$ and sufficiently large $n$ there exist an  $(n, J_n)$
code $\mathcal{C} = \bigl(E, \{D_j^n : j = 1,\cdots J_n\}\bigr)$  such that $\frac{\log
J_n}{n}
> R-\delta$, and
\begin{equation} \max_{t^n \in \theta^n} P_e(\mathcal{C}, t^n) < \epsilon\text{ ,}\label{annian1}\end{equation}
\begin{equation}\max_{t^n\in\theta^n}
\chi\left(R_{uni};Z_{t^n}\right) < \zeta\text{
,}\label{b40}\end{equation} where $R_{uni}$ is the uniform
distribution on $\{1,\cdots J_n\}$. Here
$P_e(\mathcal{C}, t^n)$ (the average probability of the decoding error of a
deterministic code $\mathcal{C}$, when the channel state  of the
arbitrarily varying classical-quantum wiretap channel $\{(W_t,{V}_t): t \in \theta\}$  is $t^n =
(t_1, t_2, \cdots , t_n)$), is defined as \[ P_e(\mathcal{C}, t^n) := 1- \frac{1}{J_n} \sum_{j=1}^{J_n}
\mathrm{tr}(W_{t^n}(E(~|j))D_j)\text{ ,}\] 
$Z_{t^n}=\Bigl\{{V}_{t^n}(E(~|i)):$ 
$i\in\{1,\cdots,J_n\}\Bigr\}$ 
 is the set of the resulting
quantum state at the output of the wiretap channel when the channel
state of $\{(W_t,{V}_t): t \in \theta\}$ is $t^n$.

 The supremum on achievable secrecy rate  for the   
$\{(W_t,{V}_t): t \in \theta\}$ is called 
the secrecy capacity of $\{({W}_t,{V}_t): t \in \theta\}$
denoted by  $C_s(\{({W}_t,{V}_t): t \in \theta\})$.
\label{defofrate}
\end{definition}

\begin{remark} A weaker and widely used
 security criterion is obtained if we replace
(\ref{b40})  with $\max_{t \in \theta}\frac{1}{n}
\chi\left(R_{uni};Z_{t^n}\right)$ $<$ $\zeta$.
\label{remsc}\end{remark}

\begin{remark}When we defined  $W_{t}$ as
$ \mathbf{A}\rightarrow\mathcal{S}(H)$ then
$P_e(\mathcal{C}, t^n)$ is defined as $1-$ $\frac{1}{J_n} \sum_{j=1}^{J_n}  \sum_{a^n \in {\mathbf{A}}^n}$
$E(a^n|j)\mathrm{tr}(W_{t^n}(a^n)D_j)$.

When deterministic encoder is used then
$P_e(\mathcal{C}, t^n)$ is defined as $1-$ $\frac{1}{J_n} \sum_{j=1}^{J_n} $
$\mathrm{tr}(W_{t^n}(c_j)D_j)$.
\end{remark}

\begin{definition}
We denote  the set of $(n, J_n)$ deterministic  codes by $\Lambda$.
A non-negative number $R$ is an achievable \bf secrecy rate \it for the
arbitrarily varying classical-quantum wiretap channel
$\{(W_t,{V}_t): t \in \theta\}$ \bf under randomness  assisted  coding \it if  for
every
 $\delta>0$, $\zeta>0$, and  $\epsilon>0$, if $n$ is sufficiently large,
there is an s a distribution $G$ on
$\left(\Lambda,\sigma\right)$ such that
$\frac{\log J_n}{n} > R-\delta$, and
\[ \max_{t^n\in\theta^n}\int_{\Lambda}P_{e}(\mathcal{C}^{\gamma},t^n)d
G(\gamma) < \epsilon\text{ ,}\]
\[\max_{t^n\in\theta^n} \int_{\Lambda}
\chi\left(R_{uni},Z_{\mathcal{C}^{\gamma},t^n}\right)dG(\gamma) < \zeta\text{
.}\] 
Here  $\sigma$ is a sigma-algebra so
chosen such that the functions $\gamma \rightarrow P_e(\mathcal{C}^{\gamma},
t^n)$ and  $\gamma \rightarrow \chi
\left(R_{uni};Z_{\mathcal{C}^{\gamma},t^n}\right)$ are both $G$-measurable
with respect to $\sigma$ for every $t^n\in\theta^n$

A non-negative number $R$ is an achievable \bf secrecy rate \it for the
arbitrarily varying classical-quantum wiretap channel
$\{(W_t,{V}_t): t \in \theta\}$ \bf under non-secure randomness  assisted  coding \it if  for
every
 $\delta>0$, $\zeta>0$, and  $\epsilon>0$, if $n$ is sufficiently large,
there is an s a distribution $G$ on
$\left(\Lambda,\sigma\right)$ such that
$\frac{\log J_n}{n} > R-\delta$, and
\[ \int_{\Lambda}\max_{t^n\in\theta^n}P_{e}(\mathcal{C}^{\gamma},t^n)d
G(\gamma) < \epsilon\text{ ,}\]
\[ \int_{\Lambda}\max_{t^n\in\theta^n}
\chi\left(R_{uni},Z_{\mathcal{C}^{\gamma},t^n}\right)dG(\gamma) < \zeta\text{
.}\] 
Here  $\sigma$ is a sigma-algebra so
chosen such that the functions $\gamma \rightarrow \max_{t^n\in\theta^n}P_e(\mathcal{C}^{\gamma},
t^n)$ and  $\gamma \rightarrow \max_{t^n\in\theta^n}\chi
\left(R_{uni};Z_{\mathcal{C}^{\gamma},t^n}\right)$ are both $G$-measurable
with respect to $\sigma$.

 The supremum on achievable secrecy rate  for the   
$\{(\mathsf{W}_t,{V}_t): t \in \theta\}$ under randomness  assisted  coding
is called 
the randomness  assisted secrecy capacity of $\{(\mathsf{W}_t,{V}_t): t \in \theta\}$
denoted by  $C_s(\{(\mathsf{W}_t,{V}_t): t \in \theta\},r)$.
The supremum on achievable secrecy rate  for the   
$\{(\mathsf{W}_t,{V}_t): t \in \theta\}$ under non-secure randomness  assisted  coding
is called 
the non-secure randomness  assisted secrecy capacity of $\{(\mathsf{W}_t,{V}_t): t \in \theta\}$
denoted by  $C_s(\{(\mathsf{W}_t,{V}_t): t \in \theta\},r_{ns})$.
\label{wdtsonjdc} 
\end{definition}

\begin{definition}
 An $(n, J_n)$    code $\mathcal{C}$ for the
 classic arbitrarily varying quantum wiretap channel $\{(\mathsf{W}_t,{V}_t): t \in \theta\}$
consists of a stochastic encoder $E$ : $\{
1,\cdots ,J_n\} \rightarrow P({\mathbf{A}}^n)$, 
$j\rightarrow E(\cdot|j)$,
 specified by
a matrix of conditional probabilities $E(\cdot|\cdot)$, and
  a collection of mutually disjoint sets $\left\{D_j
\subset \mathbf{B}^n: j\in \{ 1,\dots ,J_n\}\right\}$ (decoding sets).
\end{definition}

\begin{definition}
A non-negative number $R$ is an achievable  secrecy
rate for the  classic arbitrarily varying quantum wiretap channel 
$\{(\mathsf{W}_t,{V}_t): t \in \theta\}$ if for every $\epsilon>0$, $\delta>0$,
$\zeta>0$ and sufficiently large $n$ there exist an  $(n, J_n)$
code $\mathcal{C} = \bigl(E, \{D_j : j = 1,\cdots J_n\}\bigr)$  such that $\frac{\log
J_n}{n}
> R-\delta$, and
\begin{equation} \label{b3t} \frac{1}{J_n}\max_{t \in
\theta} \sum_{j=1}^{J_n} 
\mathsf{W}_{t}^{n}(D_j^c|E(\cdot |j))\leq  \varepsilon\text{ ,}\end{equation} and
\begin{equation}\max_{t^n\in\theta^n}
\chi\left(R_{uni};Z_{t^n}\right) < \zeta\text{ .}\label{b40t}\end{equation} 

 The supremum on achievable secrecy rate  for the   
$\{(\mathsf{W}_t,{V}_t): t \in \theta\}$ is called 
the secrecy capacity of $\{(\mathsf{W}_t,{V}_t): t \in \theta\}$
denoted by  $C_s(\{(\mathsf{W}_t,{V}_t): t \in \theta\})$.
\label{defofrate2t}
\end{definition}

\begin{definition}A non-negative number $R$ is an achievable secrecy rate for the
arbitrarily varying classical-quantum wiretap channel
$\{(W_t,{V}_t): t \in \theta\}$  under common randomness assisted  quantum coding 
using an amount $g_n$
of secret common randomness, where $g_n$ is a non-negative number depending on $n$, if  for
every
 $\delta>0$, $\zeta>0$, $\epsilon>0$, and   sufficiently large $n$,
there is a set of $(n, J_n)$   codes $(\{\mathcal{C}^{\gamma}:\gamma\in \Gamma_n\})$ 
 such that
$\frac{1}{n}\log |\Gamma_n|=g_n$,
$\frac{\log J_n}{n} > R-\delta$, and
\[ \max_{t^n\in\theta^n} \frac{1}{2^{ng_n}} \sum_{\gamma=1}^{2^{ng_n}}P_{e}(\mathcal{C}^{\gamma},t^n) < \epsilon\text{ ,}\]
\[\max_{t^n\in\theta^n} \chi\left(R_{uni},Z_{t^n}\right) < \zeta\text{ 
,}\] where $R_{uni}$ is the uniform
distribution on $\{1,\cdots J_n\}$.

Unlike in \cite{Bo/Ca/De} 
and \cite{Bo/Ca/De2} 
 we require that the randomness  to be secure
against eavesdropping here.

  The
supremum on achievable secrecy rate   under  random   assisted  quantum
coding using an amount $g_n$
of common randomness of $\{(W_t,{V}_t): t \in \theta\}$ is called the
secret random   assisted   secrecy capacity of
$\{(W_t,{V}_t): t \in \theta\}$ using an amount $g_n$
of common randomness, denoted by
$C_{key}(\{(W_t,{V}_t): t \in \theta\};g_n)$.\end{definition}

\begin{definition}

We say  super-activation occurs to
two  arbitrarily varying classical-quantum wiretap  channels
$\{({W}_t,{V}_t): t \in \theta\}$ and 
 $\{({W'}_t,{V'}_t): t \in \theta\}$ when
the following hold:
\[C_s(\{({W}_t,{V}_{t}): t\in \theta\})=0\text{ ,}\]
\[C_s(\{({W'}_t,{V'}_t): t\in \theta\})=0\text{ ,}\]
and
\[C_s(\{W_t\otimes {W'}_{t'},{V}_t\otimes {V'}_{t'}: t, t'  \in \theta\})>0\text{ .}\]

\end{definition}

\begin{remark}
For  super-activation we do not require the
strong code concept.
\end{remark}

\section{Strong Code Concept}

\subsection{Classic Arbitrarily Varying Quantum Wiretap Channel}
\label{CAVQW}

At first we determine a capacity
formula for a mixed channel model, i.e.
the secrecy capacity of   classic arbitrarily varying quantum wiretap channel.
This formula
will be used for
our  result for   secrecy capacity of
 arbitrarily
varying classical-quantum wiretap  channels using secretly sent common randomness.

\begin{theorem}
	Let $\{(\grave{W}_{t}, V_{t}): t\in \theta\}$ be a
classic arbitrarily varying quantum wiretap 
channel.
When $\{\grave{W}_{t} : t\in \theta\}$
is not symmetrizable, 
then
\[C_s(\{(\grave{W}_{t}, V_{t}): t\in \theta\})= \lim_{n\rightarrow \infty} \frac{1}{n} \max_{U\rightarrow A \rightarrow \{B_q^{n},Z_{t^n}:q,t_n\}}
\min_{q \in P({\theta})}  I(p_U,\grave{B}_q^n) -\max_{t^{n}\in \theta^{n}}\chi (p_U,Z_{t^{n}})\text{ .}\]
Here $\grave{B}_t$ are the resulting classical random variables at the output of the
legitimate receiver's channels and $Z_{t^n}$ are the resulting quantum states at the output of wiretap channels. The maximum is taken over all random
variables  that satisfy the Markov chain relationships:
$U\rightarrow A \rightarrow \grave{B}_q Z_{t}$ for every $\grave{B}_q \in Conv((\grave{B}_t)_{t\in \theta})$ and
$ t\in \theta$.  $A$ is here a random
variable taking values on $\mathbf{A}$, $U$ a random
variable taking values on some finite set $\mathbf{U}$
with probability  distribution $p_U$.
\label{lgwtvttitbaca}\end{theorem}
\begin{proof} We fix a probability distribution
$p\in P(\mathbf{A})$
and choose an arbitrarily positive $\delta$.
Let 
\[J_n =  \lfloor  2^{\inf_{\grave{B}_q \in Conv((\grave{B}_s)_{s\in \theta})}I(p;\grave{B}_q^n)-\max_{t^n\in \theta^n}\chi(p;Z_{t^n})-n\delta} \rfloor  \text{ ,}\]
and
\[L_{n} =  	\lceil 2^{\max_{t^n\in \theta^n} (\chi(p;Z_{t^n})+n\delta)}  	\text{ .}\rceil\]
Let $p' (x^n):= \begin{cases} \frac{p^{ n}(x^n)}{p^{ n}
(\mathcal{T}^n_{p,\delta})} & \text{if } x^n \in \mathcal{T}^n_{p,\delta}\text{ ;}\\
0 & \text{else .} \end{cases}$\\
 and $X^n := \{X_{j,l}\}_{j \in
\{1, \cdots, J_n\}, l \in \{1, \cdots, L_{n}\}}$  be a family of
random matrices whose components are i.i.d. according to
$p'$.

We fix a $t^n$ $\in\theta^n$ and
 define a map  $\mathsf{V}:$ $P(\theta) \times P(\mathbf{A})$
$\rightarrow$ $\mathcal{S}(H)$ by
\[\mathsf{V}(t,p):= V_{t}(p)\text{ .}\]
For  $t\in \theta$ we define $q(t)$ $:=$ $\frac{N(t\mid t^n)}{n}$.  $t^n$ 
is trivially a typical sequence of $q$.
For  $p\in  P(\mathbf{A})$,  $\mathsf{V}$
defines a map
$\mathsf{V}(\cdot,p):$
$P(\theta) $
$\rightarrow$ $\mathcal{S}(H)$.

 Let
\[Q_{t^n}(x^n) := \Pi_{\mathsf{V}(\cdot,p), \alpha }(t^n)\Pi_{\mathsf{V},\alpha}(t^n, x^n)
 \cdot V_{{t^n}}(x^n) \cdot \Pi_{\mathsf{V},\alpha}(t^n,x^n)\Pi_{\mathsf{V}(\cdot,p), \alpha }(t^n)\text{ .}\]\vspace{0.2cm}

\begin{lemma} [Gentle Operator, cf.  \cite{Win} and \cite{Og/Na}] \label{eq_4a}  
Let $\rho$ be a  quantum state and $X$ be a positive operator with $X  \leq
\mathrm{I}$  and $1 - \mathrm{tr}(\rho X)  \leq
\lambda \leq1$. Then
\begin{equation} \| \rho -\sqrt{X}\rho \sqrt{X}\|_1 \leq \sqrt{2\lambda}\text{ .} \label{tenderoper}
\end{equation}\end{lemma}

The Gentle Operator was first introduced in \cite{Win}, where it has been
shown that $\| \rho -\sqrt{X}\rho \sqrt{X}\|_1 \leq
\sqrt{8\lambda}$. In \cite{Og/Na}, the result of \cite{Win}
has been improved, and (\ref{tenderoper}) has been proved.

In view of the fact
that $\Pi_{\mathsf{V}(\cdot,p), \alpha }(t^n)$
 and $\Pi_{\mathsf{V},\alpha}(t^n, x^n)$ are
both projection matrices,
by   (\ref{te1}), (\ref{te7}),  and   Lemma \ref{eq_4a}
for any $t$ and $x^n$, it holds that
\begin{equation} \label{eq_42w}\|Q_{{t^n}}(x^n)-V_{{t^n}}(x^n)\|_1 \leq
\sqrt{2^{-n\beta(\alpha)}+2^{-n\beta(\alpha)''} }\text{ .}\end{equation}\vspace{0.2cm}

The following Lemma has been showed 
in \cite{Bo/Ca/De2}:

\begin{lemma}[Alternative Covering Lemma] \label{cov3}
Let  $\mathcal{V}$ be a finite-dimensional Hilbert space. 
Let $\mathsf{M}$ and  $\mathsf{M}'\subset \mathsf{M}$ be  finite sets. Suppose we have
an ensemble  $\{\rho_m:m\in\mathsf{M} \} \subset \mathcal{S}(\mathcal{V})$   of quantum states.
Let  $p$ be
 a   probability distribution on $\mathsf{M}$.

Suppose a total subspace projector $\Pi$ and codeword subspace projectors $\{\Pi_{m}: m\in \mathsf{M}\}$  exist  which project
onto subspaces of the Hilbert space in which the states exist, and  for all $ m\in \mathsf{M}'$ there are positive constants
$\epsilon\in
]0,1[$, $D$, $d$ such that  the following conditions hold:
\[\mathrm{tr}(\rho_m\Pi)\geq 1-\epsilon\text{ ,}\] \[\mathrm{tr}(\rho_m\Pi_m)\geq 1-\epsilon\text{ ,}\]
\[\mathrm{tr}(\Pi)\leq D\text{ ,}\] 
and 
\[\Pi_m\rho_m\Pi_m\leq\frac{1}{d}\Pi_m\text{ .}\]

We denote $\omega:=\sum_{m\in\mathsf{M}'} p(m) \rho_m$.
Notice that $\omega$ is not a
density operator in general.
We define a  sequence of  i.i.d.~random variables  $X_1,
\dots ,X_L$, taking values  
in $\{\rho_m:m\in\mathsf{M} \}$. If $L\gg \frac{d}{D}$ then

\begin{align}&
Pr \biggl( \lVert L^{-1}
\sum_{i=1}^{L}\Pi\cdot\Pi_{X_i}\cdot X_i \cdot\Pi_{X_i} \cdot\Pi
- \omega \rVert_1\allowdisplaybreaks\notag\\
&\leq    1-p(\mathsf{M}')+  4\sqrt{1-p(\mathsf{M}')}+ 
42 \sqrt[8]{\epsilon} 
 \biggr)\allowdisplaybreaks\notag\\
 & \geq 1- 2D \exp
\left( -p(\mathsf{M}')\frac{\epsilon^3Ld}{2\ln 2D} \right) \text{ .}\label{pbllsilpcpx32w}
\end{align}
\end{lemma}

By
(\ref{te2})  we have
\begin{align}&\mathrm{tr}(\Pi_{\mathsf{V}(\cdot,p), \alpha }(t^n))\allowdisplaybreaks\notag\\
&\leq 2^{n (S(\mathsf{V}(\cdot,p)\mid q) +\delta(\alpha))}\allowdisplaybreaks\notag\\
&= 2^{n (\sum_{t}q(t)\mathsf{V}(t,p) +\delta(\alpha))}\allowdisplaybreaks\notag\\
&= 2^{n (\sum_{t}q(t)S(V_{t}(p)) +\delta(\alpha))}\text{ .}\label{eq_5a22w}\end{align}

Furthermore, for all $x^n$ it holds that

\begin{align}&
\Pi_{\mathsf{V},\alpha}(t^n, x^n)
 V_{{t^n}}(x^n)  \Pi_{\mathsf{V},\alpha}(t^n,x^n)\allowdisplaybreaks\notag\\
&\leq 2^{-n(S(\mathsf{V}|r) + \delta(\alpha)')}\Pi_{\mathsf{V},\alpha}(t^n,x^n)\allowdisplaybreaks\notag\\
&= 2^{-n(\sum_{t,x}r(t,x)S(\mathsf{V}(t,x)) + \delta(\alpha)')}\Pi_{\mathsf{V},\alpha}(t^n,x^n)
\text{ .}\label{eq_5a2w}
\end{align} 
\vspace{0.15cm}

We define 
\[\theta':=\left\{t\in\theta:nq(t)\geq \sqrt{n}\right\}\text{ .}\]
By properties of classical typical set (cf. \cite{Win} ) there is a positive $\hat{\beta}(\alpha)$
such that
\begin{equation} \underset{p'}{Pr} \biggl(x^n\in \biggl\{x^n\in{\mathbf{A}}^n:
(x_{\mathtt{I}_t})\in \mathsf{T}^{nq(t)}_{p,\delta} ~ \forall 
t\in\theta'\biggr\}\biggr)\geq \left(1-2^{-\sqrt{n}\hat{\beta}(\alpha)}\right)^{|\theta|}\geq 1-2^{-\sqrt{n}
\frac{1}{2}\hat{\beta}(\alpha)}\text{ ,}\label{uppbxm2w}\end{equation} 
where $\mathtt{I}_t$ $:=$ $\{i\in\{1,\cdots,n\}: t_i = t\}$ 
is an indicator set that selects the indices $i$ in the sequence $t^n$
$=$ $(t_1,\cdots,t_n)$.

We denote the set $\{x^n: (x_{\mathtt{I}_t})\in \mathsf{T}^{nq(t)}_{p,\delta} ~ \forall 
t\in\theta'\}$ $\subset {\mathbf{A}}^n$ by $\mathsf{M}_{t^n}$.
For all $x^n\in\mathsf{M}_{t^n}$, if $n$ is sufficiently
large,
we have
\begin{align}&\left\vert\sum_{t,x}r(t,x)S(\mathsf{V}(t,x))-
\sum_{t}q(t)S(V_{t}|p)\right\vert\allowdisplaybreaks\notag\\
&\leq\left\vert\sum_{t\in\theta',x}r(t,x)S(\mathsf{V}(t,x))-
\sum_{t\in\theta'}q(t)S(V_{t}|p)\right\vert\allowdisplaybreaks\notag\\ 
&+ \left\vert\sum_{t\notin\theta',x}r(t,x)S(\mathsf{V}(t,x))-
\sum_{t\notin\theta'}q(t)S(V_{t}|p)\right\vert\allowdisplaybreaks\notag\\
&\leq\sum_{t\in\theta'}\left\vert\sum_{x}r(t,x)S(\mathsf{V}(t,x))-
q(t)S(V_{t}|p)\right\vert+2|\theta|\frac{1}{\sqrt{n}}C \notag\\ 
&\leq 2|\theta|\frac{\delta}{n}C  +2|\theta|\frac{1}{\sqrt{n}}C
\text{ ,}\label{uppbxmn22w} \end{align} where $C:=\max_{t\in\theta}\max_{x\in \mathbf{A}}(S(\mathsf{V}(t,x))+
S(V_{t}|p))$.
We set $\Theta_{t^n}:=  \sum_{x^n \in \mathsf{M}_{t^n}}
{p}(x^n) Q_{{t^n}}(x^n)$. For given $z^n\in\mathsf{M}_{t^n}$ and ${t^n}\in\theta^n$, $\langle
z^n|\Theta_{t^n}|z^n\rangle$  is the expected  value of $\langle z^n|
Q_{{t^n}}(x^n) |z^n\rangle$ under the condition $x^n \in
\mathsf{M}_{t^n}$.\vspace{0.15cm}

We choose a positive $\bar{\beta}(\alpha)$
such that  $\bar{\beta}(\alpha)\leq \min(2^{-n\beta(\alpha)},
2^{-n\beta(\alpha)'})$, and 
set
$\epsilon := 2^{-n\bar{\beta}(\alpha)}$.
In view of (\ref{eq_5a2w}) 
we now apply Lemma \ref{cov3}, where
we consider the set $\mathsf{M}_{t^n}$ $\subset {\mathbf{A}}^n$:
If  $n$ is sufficiently large, for all $j$ 
 we have

\begin{align}&Pr \left( \lVert \sum_{l=1}^{L_{n}} \frac{1}{L_{n}} Q_{t^n}(X_{j,l}) -
\Theta_{t^n} \rVert_1 >  2^{-\sqrt{n}
\frac{1}{8}\hat{\beta}(\alpha)}  +40 \sqrt[8]{\epsilon} 
\right)\allowdisplaybreaks\notag\\
&\leq 2^{n (\sum_{t,x}r(t,x)S(\mathsf{V}(t,x)) +\delta(\alpha))} \notag\\
&\cdot \exp \left( -L_{n}\frac{\epsilon^3}{2\ln  2}  (1-2^{-\sqrt{n}\frac{1}{2}\hat{\beta}(\alpha)})
 \cdot 2^{n(\sum_{t}q(t)S(V_{t}(p))-\sum_{t}q(t)S(V_{t}|p)) +\delta(\alpha)+ \delta(\alpha)' + 2|\theta|\frac{\delta}{n}C
+ 2|\theta|\frac{1}{\sqrt{n}}C}
 \right)\allowdisplaybreaks\notag\\
&= 2^{n (\sum_{t,x}r(t,x)S(\mathsf{V}(t,x)) +\delta(\alpha)} \notag\\
&\cdot \exp \left( -L_{n}\frac{\epsilon^3}{2\ln  2}  \cdot (1-2^{-\sqrt{n}\frac{1}{2}\hat{\beta}(\alpha)})
2^{n (-\sum_{t}q(t)\chi(p;Z_{t}) + \delta(\alpha)+\delta(\alpha)'+ 2|\theta|\frac{\delta}{n}C
+ 2|\theta|\frac{1}{\sqrt{n}}C )} \right)\text{
.}\label{eq_5b+2w}\end{align}  The  equality  holds since
$S(V_{t}(p))-S(V_{t}|p)  =\chi(p;Z_{t})$. 

Furthermore,
 \begin{align}&Pr \left( \lVert \sum_{l=1}^{L_{n}} \frac{1}{L_{n}} Q_{t^n}(X_{j,l}) -
\Theta_{t^n} \rVert_1 >   2^{-\sqrt{n}
\frac{1}{8}\hat{\beta}(\alpha)}  +40 \sqrt[8]{\epsilon} ~ \forall t^n ~ \forall j \right)\allowdisplaybreaks\notag\\
&\leq J_n {|\theta|}^n 2^{n (\sum_{t,x}r(t,x)S(\mathsf{V}(t,x)) +\delta(\alpha)} \notag\\
&\cdot \exp \left( -L_{n}\frac{\epsilon^3}{2\ln  2} (1-2^{-\sqrt{n}\frac{1}{2}\hat{\beta}(\alpha)})
2^{n (-\sum_{t}q(t)\chi(p;Z_{t}) + \delta(\alpha)+\delta(\alpha)'+ 2|\theta|\frac{\delta}{n}C
+ 2|\theta|\frac{1}{\sqrt{n}}C )} \right)\text{
.}\label{eq_5b2w}\end{align} 
\vspace{0.2cm}

Let $\phi_{t}^{j}$ be the quantum state at the output of wiretapper's channel when
the channel state is $t$ and $j$ has
been sent.
We have
\begin{align*}&\sum_{t\in \theta}q(t)\chi\left(p;Z_{t}\right)- \chi\left(p;\sum_{t}q(t)Z_{t}\right)\\
&= \sum_{t\in \theta}q(t)S\left(\sum_{j=1}^{J_n}\frac{1}{J_n} \phi_{t}^{j}\right)
-\sum_{t\in \theta}\sum_{j=1}^{J_n}q(t)\frac{1}{J_n}S\left( \phi_{t}^{j}\right)\\
&-S\left(\frac{1}{J_n}\sum_{t\in \theta}\sum_{j=1}^{J_n} q(t)\phi_{t}^{j}\right)
+\sum_{j=1}^{J_n}\frac{1}{J_n}S\left(\sum_{t\in \theta}q(t) \phi_{t}^{j}\right)\text{ .}
\end{align*}
 Let $H^{\mathfrak{T}}$ be a $\left|\theta\right|$-dimensional Hilbert space
spanned by an orthonormal basis $\{|t\rangle : t = 1, \cdots, \left|\theta\right|\}$. 
Let $H^{\mathfrak{J}}$ be a $J_n$-dimensional Hilbert space spanned by an orthonormal basis 
$\{|j\rangle : j = 1, \cdots, J_n\}$. 
We define
\[\varphi^{\mathfrak{J}\mathfrak{T}H^{n}}:=\frac{1}{J_n}\sum_{j=1}^{J_n}\sum_{t\in \theta}q(t)
|j\rangle\langle j|\otimes|t\rangle\langle t|\otimes
 \phi_{t}^{j}\text{ .}\]
We have
\[\varphi^{\mathfrak{J}H^{n}}=\mathrm{tr}_{\mathfrak{T}}\left(\varphi^{\mathfrak{J}\mathfrak{T}H^{n}}\right)=
\frac{1}{J_n}\sum_{j=1}^{J_n}\sum_{t\in \theta}q(t)
|j\rangle\langle j|\otimes
 \phi_{t}^{j}\text{ ;}\]
\[\varphi^{\mathfrak{T}H^{n}}=\mathrm{tr}_{\mathfrak{J}}\left(\varphi^{\mathfrak{J}\mathfrak{T}H^{n}}\right)=
\frac{1}{J_n}\sum_{j=1}^{J_n}\sum_{t\in \theta}q(t)|t\rangle\langle t|\otimes
 \phi_{t}^{j}\text{ ;}\]
\[\varphi^{H^{n}}=\mathrm{tr}_\mathfrak{J}{\mathfrak{T}}\left(\varphi^{\mathfrak{J}\mathfrak{T}H^{n}}\right)=
\frac{1}{J_n}\sum_{j=1}^{J_n}\sum_{t\in \theta}q(t)
 \phi_{t}^{j}\text{ .}\]
Thus,
\[S(\varphi^{\mathfrak{J}H^{n}}) = H(R_{uni})+ \frac{1}{J_n}\sum_{j=1}^{J_n}S\left(\sum_{t\in \theta}q(t)
\phi_{t}^{j}\right)\text{ ;}\]
\[S(\varphi^{\mathfrak{T}H^{n}})=H(Y_q)+  \sum_{t\in \theta}q(t)S\left(\frac{1}{J_n}\sum_{j=1}^{J_n}
 \phi_{t}^{j}\right)\text{ ;}\]
\[S(\varphi^{\mathfrak{J}\mathfrak{T}H^{n}})= H(R_{uni})+H(Y_q)+ \frac{1}{J_n}\sum_{j=1}^{J_n}\sum_{t\in \theta}q(t)S\left(
 \phi_{t}^{j}\right)\text{ ,}\]
where $Y_q$ is a random variable on $\theta$ with distribution $q(t)$.

 By strong subadditivity of von Neumann entropy, it holds that $S(\varphi^{\mathfrak{J}H^{n}}) + S(\varphi^{\mathfrak{T}H^{n}})$
$\geq$ $S(\varphi^{H^{n}})+S(\varphi^{\mathfrak{j}\mathfrak{T}H^{n}})$, therefore
\begin{equation}\sum_{t}q(t)\chi\left(p;Z_{t}\right)- \chi\left(p;\sum_{t}q(t)Z_{t}\right)
\geq 0\text{ .}\label{stqtclpztr2w}\end{equation}\vspace{0.2cm}

For an arbitrary $\zeta$,
we define $L_{n} = \lceil 2^{\max_{t^n}\chi(p;Z_{t^n})+n\zeta} \rceil$,
and choose a suitable $\alpha$, $\bar{\beta}(\alpha)$, and 
sufficiently large $n$ such that
$6\bar{\beta}(\alpha)$ + $2\delta(\alpha)$ $+2\delta(\alpha)'$ $+ 2|\theta|\frac{\delta}{n}C$
$+ 2|\theta|\frac{1}{\sqrt{n}}C$ $\leq\zeta$.
By (\ref{stqtclpztr2w}), if $n$ is  sufficiently large, we have
$L_{n} \geq \lceil 2^{n(\sum_{t}q(t)\chi(p;Z_{t})+\zeta)} \rceil$
and
\begin{align*}&
L_{n}\frac{\epsilon^3}{2\ln  2} (1-2^{-\sqrt{n}\frac{1}{2}\hat{\beta}(\alpha)})
2^{n (-\sum_{t}q(t)\chi(p;Z_{t}) + \delta(\alpha)+\delta(\alpha)'+ 2|\theta|\frac{\delta}{n}C
+ 2|\theta|\frac{1}{\sqrt{n}}C )} > 2^{\frac{1}{2}n\zeta}\text{ .}
\end{align*} 

When $n$ is sufficiently large 
for any positive $\vartheta$ it holds that
\begin{align*}&J_n{|\theta|}^n 2^{n (\sum_{t,x}r(t,x)S(\mathsf{V}(t,x)) +\delta(\alpha)} \exp(- 2^{\frac{1}{4}n\zeta})\\
&\leq 2^{-n\vartheta}
\end{align*}
and 
\[ 2^{-\sqrt{n}
\frac{1}{8}\hat{\beta}(\alpha)}  +40 \sqrt[8]{\epsilon}  \leq 2^{-\sqrt{n}
\frac{1}{16}\hat{\beta}(\alpha)} \text{ .}\]

Thus for sufficiently large $n$ 
we have
\begin{align} &Pr \biggl( \lVert \sum_{l=1}^{L_{n}} \frac{1}{L_{n}} Q_{t^n}(X_{j,l}) -\Theta_{t^n}
\rVert_1 \leq  2^{-\sqrt{n}
\frac{1}{16}\hat{\beta}(\alpha)} \text{ }\forall t^n  ~ \forall j 
 \biggr)\allowdisplaybreaks\notag\\
&\geq 1- 2^{n\vartheta}
\label{eq_6blslln2w}
\end{align} for any positive $\varrho $.
\vspace{0.3cm}

In 
\cite{Cs/Na} , the following was
 shown: Let $\{{X}_{j,l}\}_{j \in \{1, \dots, J_n\}, l \in
\{1, \dots, L_{n}\}}$ be a family of random variables taking value
according to $p'$. We assume $ \{\grave{W}_{t}: t\in \theta\}$ is not symmetrizable.
If $n$ is sufficiently large,
and if $J_n\cdot L_{n} \leq 2^{inf_{\grave{B}_q \in Conv((\grave{B}_t)_{t\in \theta})}I(p;\grave{B}_q)-\mu)}$
for an  arbitrary positive $\mu$
there exists a set  of mutually disjoint sets 
$\{D_{j,l}: j \in
\{1, \cdots, J_n\}, l \in \{1, \cdots, L_{n}\}$ on $\mathbf{B}^n$ such that
 for all positive $\epsilon$,  $t^n \in \theta^n$,
 and $j \in \{1,\dots,J_n\}$
\begin{equation}\underset{p'}{Pr} \left[\mathrm{tr}\left(\grave{W}_{t^n}({X}_{j,l})D_{j,l}\right) \geq 
1-2^{-{n}\beta}\right] > 1-2^{-n\gamma } \text{
,}\label{qnocsig4'2w}\end{equation}

By (\ref{eq_6blslln2w}) and (\ref{qnocsig4'2w}), when $ \{\grave{W}_{t}: t\in \theta\}$ is not symmetrizable
  we can find with positive probability a realization $x_{j,l}$ of
$X_{j,l}$ 
and
 set  of mutually disjoint sets 
$\{D_{j,l}: j \in
\{1, \cdots, J_n\}, l \in \{1, \cdots, L_{n}\}$ such that
 for all positive $\epsilon$,  $t^n \in \theta^n$,
 and $j \in \{1,\dots,J_n\}$
\begin{equation} \label{b3lnrimti} \max_{t \in
\theta} \frac{1}{J_n} \sum_{j=1}^{J_n}
\grave{W}_{t^n}(D_{j,l}^c| x_{j,l})\leq \epsilon \text{ ,}\end{equation}
 and 
\begin{equation} 
\lVert \sum_{l=1}^{L_{n}} \frac{1}{L_{n}} Q_{t^n}(x_{j,l}) -\Theta_{t^n}
\rVert_1\leq \epsilon\label{elrttljx}
\text{ .}\end{equation}\vspace{0.2cm}
Here we define $E(x^n\mid j)=\frac{1}{L_{n}}$ if $x^n\in \{x_{j,l} :l\in\{1,\dots,L_{n}\}$,

We choose a suitable positive  $\alpha$.
For any given $j' \in \{1, \dots, J_n\}$,
by (\ref{eq_42w}) and (\ref{elrttljx}) we have
\begin{align}
& \left\|\sum_{l=1}^{L_n}\frac{1}{L_n}V_{t^n}(x_{j',l})-\Theta_{t^n}\right\|_1\notag\\
 &\leq\|\sum_{l=1}^{L_n}\frac{1}{L_n}V_{t^n}(x_{j',l})-\sum_{l=1}^{L_n}\frac{1}{L_n}Q_{t^n}(x_{j',l})\|_1\notag\\
 &\qquad+\|\sum_{l=1}^{L_n}\frac{1}{L_n}Q_{t^n}(x_{j',l})-\Theta_{t^n}\|_1\notag\\
&\leq 2^{-\sqrt{n}
\frac{1}{16}\hat{\beta}(\alpha)}+\sqrt{2^{-\frac{1}{2}n\beta(\alpha)}+2^{-\frac{1}{2}n\beta(\alpha)''}}\notag\\
 &\leq 2^{-\sqrt{n}
\frac{1}{32}\hat{\beta}(\alpha)}\text{ .}\label{eq_82w}
 \end{align}
Notice that by (\ref{eq_82w}) we have
$\|\frac{1}{J_n\cdot L_n}\sum_{j=1}^{J_n}\sum_{l=1}^{L_n}V_{t^n}(x_{j,l})-\Theta_{t^n}\|_1$
$\leq  2^{-\sqrt{n}
\frac{1}{32}\hat{\beta}(\alpha)}$.\vspace{0.2cm}

\begin{lemma}[Fannes-Audenaert  Ineq.,
 cf. \cite{Fa}, \cite{Au}]\label{eq_9}  
Let $\Phi$ and $\Psi$ be two  quantum states in a
$d$-dimensional complex Hilbert space and
$\|\Phi-\Psi\| \leq \mu < \frac{1}{e}$, then
\begin{equation} |S(\Phi)-S(\Psi)| \leq \mu \log (d-1)
+ h(\mu)\text{ ,}\label{faaudin}\end{equation}where $h(\nu) := -\nu \log \nu - (1- \nu) \log (1-\nu)$
for $\nu\in [0,1]$.\end{lemma}\vspace{0.15cm}

 The Fannes  Inequality was first introduced in \cite{Fa}, where it has been
shown that $|S(\mathfrak{X})-S(\mathfrak{Y})| \leq \mu \log d - \mu
\log \mu $. In \cite{Au} the result of \cite{Fa} has been
improved, and (\ref{faaudin}) has been proved.\vspace{0.15cm}

By Lemma \ref{eq_9}  and  the inequality (\ref{eq_82w}),  for a uniformly distributed
 random variable $R_{uni}$ with values  in
$\{1,\dots,J_n\}$ a and $t^n\in\theta^n$, we have
\begin{align}& \chi(R_{uni};Z_{t^n}) \allowdisplaybreaks\notag\\
&=S\left( \sum_{j=1}^{J_n} \frac{1}{J_{n}} \sum_{l=1}^{L_{n}}
\frac{1}{L_{n}} V_{t^n}(x_{j,l})\right) \allowdisplaybreaks\notag\\
&- \sum_{j=1}^{J_n}
\frac{1}{J_{n}}S\left(\sum_{l=1}^{L_{n}}
 \frac{1}{L_{n}}V_{t^n}(\pi(x_{j,l}))\right)\allowdisplaybreaks\notag\\
  &\leq \left\vert  S\left( \sum_{j=1}^{J_n} \frac{1}{J_{n}} \sum_{l=1}^{L_{n}}
\frac{1}{L_{n}} V_{t^n}(x_{j,l})\right)-S\left( \Theta_{t^n} \right) \right\vert \allowdisplaybreaks\notag\\
&+\left\vert  S(\Theta_{t^n} )- \sum_{j=1}^{J_n} \frac{1}{J_{n}}S\left( \sum_{l=1}^{L_{n}}
\frac{1}{L_{n}} V_{t^n}(x_{j,l})\right)\right\vert \allowdisplaybreaks\notag\\
&\leq 2\cdot  2^{-\sqrt{n}
\frac{1}{32}\hat{\beta}(\alpha)} \log (nd-1) + 
2h( 2^{-\sqrt{n}
\frac{1}{32}\hat{\beta}(\alpha)}) 
\text{ .}\label{wehaveb2w}\end{align}

By (\ref{wehaveb2w}), for any positive $\lambda$ if $n$ is sufficiently
large, we have
\begin{equation} \label{clruztr2w}
\chi\left(R_{uni};Z_{t^n}\right)\leq \lambda
\text{ .}\end{equation}

We define $E(x^n\mid j)=\begin{cases}
 \frac{1}{L_{n}} \text{ if }x^n\in \{x_{j,l} :l\in\{1,\dots,L_{n}\}\}\text{ ;}\\
    0 \text{ if }  x\not\in \{x_{j,l} :l\in\{1,\dots,L_{n}\}\}     \text{ .}     \end{cases}
$ and $D_j := \bigcup_{l} D_{j,l}$.
By (\ref{elrttljx}) and (\ref{clruztr2w}), when $ \{\grave{W}_{t}: t\in \theta\}$ is not symmetrizable
the 
deterministic secrecy capacity
of 
$\{(\grave{W}_{t}, V_{t}): t\in \theta\}$
is larger or  equal to 

\begin{equation}
\lim_{n\rightarrow \infty} \frac{1}{n}\left(\inf_{B_q \in Conv((B_s)_{s\in \theta})}\chi(p;\grave{B}_{q}^n)
-\max_{t^n\in \theta^n}\chi(p;Z_{t^n})\right)-\varepsilon
	\text{ .}\label{hatceque3}\end{equation}

The achievability of
$\lim_{n\rightarrow \infty}$ $\frac{1}{n}$
$\max_{U\rightarrow A \rightarrow \{B_q,Z_{t}:q,t\}}$ 
$(\inf_{\grave{B}_q \in Conv((\grave{B}_s)_{s\in \theta})}I(p_U;\grave{B}_q)$ $-$ $ \max_{t^n\in \theta^n}\chi(p_U;Z_{t^n}))$
and the converse are shown by the standard arguments  (cf. \cite{De} and \cite{Bj/Bo}).
\end{proof}

\subsection{The Secure Message Transmission  With Strong Code Concept}
\label{TSMwSCR}

Now we are going to prove our main result:
the  secrecy capacity formula for
arbitrarily varying  classical-quantum wiretap  channels
 using  secretly sent common randomness.
In our previous papers  \cite{Bo/Ca/De} and \cite{Bo/Ca/De2}
we determined the
secrecy capacity formula for
arbitrarily varying  classical-quantum wiretap  channels.
Our strategy is to
build a two-part code word,
which consists of a non-secure code word and a common randomness-assisted
secure code word. The non-secure one is used to create the common randomness
for the sender and the legal receiver. The common randomness-assisted secure
code word is used to transmit the message to the legal receiver.

Now we
build a code in such a way that the transmission of both
the message and the randomization is secure.
Since the technique introduced
in \cite{Cs/Na} for classical  channels cannot be easily transferred into quantum channels,
our idea is to
construct a classical arbitrarily varying quantum wiretap 
channel and apply Theorem \ref{lgwtvttitbaca}.
In \cite{Ahl/Bli}
a technique has been introduced to
construct a classical arbitrarily varying channel by means of
an arbitrarily varying classical-quantum channel.
However this technique does not work
for classical arbitrarily varying quantum wiretap 
channel since it cannot provide security.
We have to find a more sophisticated way.

\begin{theorem}	

 If the
arbitrarily varying  classical-quantum channel $\{W_t : t \in \theta\}$
is not symmetrizable, then
	\begin{equation}C_{s}(\{(W_t,{V}_t): t \in \theta\}) =
\lim_{n\rightarrow \infty} \frac{1}{n}\max_{U\rightarrow A \rightarrow \{B_q,Z_{t}:q,t\}}
	\Bigl(\inf_{B_q \in Conv((B_t)_{t\in \theta})}\chi(p_U;B_q^{\otimes n})-\max_{t^n\in \theta^n}\chi(p_U;Z_{t^n})\Bigr)
	\label{bibqicbtt} \text{ ,}\end{equation}
when we use a two-part code word
that both parts are secure.
	
Here $B_t$ are the resulting  quantum states at the output of the
legitimate receiver's channels. $Z_{t^n}$ are the resulting  quantum states  at
the output of wiretap channels.
 The maximum is taken over all random
variables  that satisfy the Markov chain relationships:
$U\rightarrow A \rightarrow B_q Z_{t}$ for every $B_q \in Conv((B_t)_{t\in \theta})$ and
$ t\in \theta$.  $A$ is here a random
variable taking values on $\mathbf{A}$, $U$ a random
variable taking values on some finite set $\mathbf{U}$
with probability  distribution $p_U$.
\label{pdwumsfsov}
\end{theorem}

\begin{proof}

Since the security of both the message and the randomization 
implies the security of only the message,
 the secrecy
capacity of $\{(W_t,{V}_t): t \in \theta\}$ for 
the message and the randomization transmission cannot exceed
$\lim_{n\rightarrow \infty}$ $\frac{1}{n}$ $\max_{U\rightarrow A \rightarrow \{B_q^{\otimes n},Z_{t^n}:q,t_n\}}$
$\inf_{B_q \in Conv((B_t)_{t\in \theta})}\chi(p_U;B_q^{\otimes n})$ $-\max_{t^n\in \theta^n}\chi(p_U;Z_{t^n})$,
which is the   secrecy
capacity of $\{(W_t,{V}_t): t \in \theta\}$ for 
only the message transmission (cf. \cite{Bo/Ca/De2}).
Thus the converse is trivial.

For the achievability we at first
assume that $\{W_t : t \in \theta\}$
is  symmetrizable. In this case
the secrecy
capacity of $\{(W_t,{V}_t): t \in \theta\}$ for 
only the message transmission is zero and there
is nothing to prove.
At next we assume that for all $p\in P(\mathbf{U})$ we have
$\lim_{n\rightarrow \infty}$ $\frac{1}{n}$ $\max_{U\rightarrow A \rightarrow \{B_q^{\otimes n},Z_{t^n}:q,t_n\}}$
$\Bigl(\inf_{B_q \in Conv((B_t)_{t\in \theta})}\chi(p_U;B_q^{\otimes n})$ $-\max_{t^n\in \theta^n}\chi(p_U;Z_{t^n})\Bigr)$
$\leq 0$. In this case
the secrecy
capacity of $\{(W_t,{V}_t): t \in \theta\}$ for 
only the message transmission is also zero and again there
is nothing to prove.
Now we assume that $\{W_t : t \in \theta\}$
is not symmetrizable and
 for all sufficiently large $n$ and a positive $\epsilon$
\begin{equation}\frac{1}{n} \max_{U\rightarrow A \rightarrow \{B_q^{\otimes n},Z_{t^n}:q,t_n\}}\left(\inf_{B_q \in Conv((B_t)_{t\in \theta})}
\chi(p_U;B_q^{\otimes n})-\max_{t^n\in \theta^n}\chi(p;Z_{t^n})\right)>2\epsilon
\label{finmpibqic1}\end{equation}
holds.
\vspace{0.2cm}

\it i) Construction of a Non-Symmetrizable Channel with Random Pre-coding \rm\vspace{0.2cm}

We consider the Markov chain
 $U\rightarrow A \rightarrow \{B_q,Z_{t}:q,t\}$, where
we define the classical channel $P(\mathbf{U})\rightarrow P(\mathbf{A})$ by $T_U$.
It may
happen that $\{{W}_t\circ T_U: t \in \theta\}$ is symmetrizable although $\{{W}_t: t \in \theta\}$ is not
symmetrizable, as following example shows:

We assume that  $\{{W}_t: t \in \theta\}$ :$P(\mathbf{A})\rightarrow \mathcal{S}(H)$
is not symmetrizable but there is a subset $\mathbf{A}'\subset \mathbf{A}$ such that
 $\{{W}_t: t \in \theta\}$ limited on $A'$ is symmetrizable. We choose
a $T_U$ such that  for every  $u\in \mathbf{U}$ there is $a\in \mathbf{A}'$
such that $T_U(a\mid u) = 1$, and 
$T_U(a\mid u) = 0$ for all $a\in \mathbf{A}\setminus \mathbf{A}'$ and $u\in \mathbf{U}$.
It is clear that $\{{W}_t\circ T_U: t \in \theta\}$ is symmetrizable
(cf. also \cite{Wi/No/Bo2} for an example for classical channels).

We now use a  technique introduced in
\cite{Wi/No/Bo2} to overcome this:
Without loss of generality we may 
assume that 
$|\mathbf{A}|=|\mathbf{U}|$ by optimization.
Furthermore
without loss of generality we may 
assume that 
$\mathbf{A}=\mathbf{U}$ by relabeling the symbols.
For every $n>1\in \mathbb{N}$ we define a new classical channel
$\tilde{T}_U^{n}$ $:P(\mathbf{A}^{n}) \rightarrow P(\mathbf{A}^{n})$
by setting $\tilde{T}_U^{n}$ $:= T_U^{n-1} \times id_{\mathbf{A}}$, i.e.,
\[\tilde{T}_U^{n}(a_1,\cdots,a_{n-1},a_{n}) := T_U^n(a_1,\cdots,a_{n-1})\cdot \delta_{a_{n}}\text{ .}\]
We have 
\begin{equation}{W}_{t^{n}}\circ \tilde{T}_U^{n}(a_1,\cdots,a_{n-1},a_{n}) = {W}_{t^{n-1}}\circ T_U^n(a_1,\cdots,a_{n-1}){W}_{t_n}(a_{n})\text{ ,}\end{equation}
where for $t^n=(t_1,\cdots,t_{n-1},t_n)$ we denote $t^{n-1}:=(t_1,\cdots,t_{n-1})$.
Since $\{{W}_t: t \in \theta\}$ is not
symmetrizable,
 $\{{W}_{t^{n}}\circ  \tilde{T}_U^{n}: t \in \theta\}$ is not symmetrizable.
Furthermore, for any positive $\delta$
sufficiently large $n$ 
we have 
\[ C(\{{W}_{t^{n}}\circ T_U^{n}: t \in \theta\}; r) \leq C(\{{W}_{t^{n-1}}\circ T_U^{n-1}: t \in \theta\}; r)+\delta
\leq  C(\{{W}_{t^{n}}\circ \tilde{T}_U^{n-1}: t \in \theta\}; r) +\delta\text{ .}\]

For every $n>1$ and $t^{n}\in {\theta}^{n}$ we define 
$\check{W}_{t^{n}}$ $:P(\mathbf{U}^{n})\rightarrow P(\mathbf{A}^{n})$ by
\[\check{W}_{t^{n}}:={W}_{t^{n}}\circ  \tilde{T}_U^{n}\text{ .}\]

This shows that, if for a non-symmetrizable channel $\{W_t:t\in\Theta\}$ 
we have (\ref{finmpibqic1}), then 
 \begin{equation}\frac{1}{n} \max_{U\rightarrow A \rightarrow \{\check{B}_q^{\otimes n},Z_{t^n}:q,t_n\}}
\left(\inf_{\check{B}_q \in Conv((\check{B}_t)_{t\in \theta})}
\chi(p_U;\check{B}_q^{\otimes n})-\max_{t^n\in \theta^n}\chi(p;Z_{t^n})\right)>\epsilon
\label{finmpibqic}\end{equation} holds,
where $\check{B}_{t^{n}}$ are the resulting  quantum states at the output of $\check{W}_{t^{n}}$.
\vspace{0.2cm}

\it ii) Definition of a Classical Arbitrarily Varying Channel Which Is  Not Symmetrizable \rm\vspace{0.2cm}

We denote $m:= \log n$  
and define $\check{V}_t :={V}_{t}\circ T_U$ for all
 $t \in \theta$.
Now we
consider the 
arbitrarily varying wiretap classical-quantum channel
$\{(\check{W}_t,\check{V}_t);  t \in \theta\}$.
We choose an arbitrary   $\delta>0$, 
by (\ref{finmpibqic}) if $m$ is sufficiently large
we may assume that for at least one $p\in P(\mathbf{U})$
\[2 \leq 2^{m\epsilon}
\leq \lfloor  2^{m\inf_{\check{B}_q \in Conv((\check{B}_s)_{s\in \theta})}\chi(p;\check{B}_q)-\max_{t^m\in \theta^m}\chi(p;Z_{t^m})-m\delta} \rfloor\text{ .}\]

By Theorem 1 of \cite{Bo/Ca/De2} if $m$ is sufficiently large we can find a 
$(m, 2)$   code $\bigl(E^m , \{D_j^m : j \in \{1, 2\}\}\bigr)$ and  positive  $\lambda$, $\zeta$
such that for some $q\in P(\theta)$ 

	\begin{equation}
1- \frac{1}{2} \sum_{j=1}^{2}
\mathrm{tr}\left( \check{W}_{q}^{\otimes m}( E^m(~|j))D_j^m \right)\geq 1-2^{-m^{1/16}\lambda}
\label{1f1jmspi}\end{equation}
and for all $t^m\in\theta^m$ and $\pi \in \Pi_m$
\begin{align}&\lVert
\check{V}_{t^m}\left(\pi(E^m(~|j))\right) -
\Theta_{t^m} \rVert
< 2^{-\sqrt{m}\zeta}
\label{balf1lnsl1}\end{align}
for a $\Theta_{t^m}\in\mathcal{S}(H^m)$ which is independent of $j$.
Here for $\pi\in\mathsf{S}_m$ we define
its permutation matrix on ${H}^{\otimes m}$ by $P_{\pi}$.

(Notice that $\bigl(E^m , \{D_j^m : j \in \{1, 2\}\}\bigr)$
is a deterministic code for a mixed channel model
called compound-arbitrarily varying wiretap classical-quantum channel
which we introduced in \cite{Bo/Ca/De2}.)
\vspace{0.2cm}

We now combine a technique introduced in
\cite{Ahl/Bli} with the concept of superposition code to
define a set of classical channels.

We choose ${d^m}^2+1$
 Hermitian operators 
$L_i\geq 0$, $i=1, \cdots, {d^m}^2+1$ which span the space of Hermitian operators on
 $H^m$ and fulfill $\sum_{i=1}^{{d^m}^2+1} L_i = id_{H^m}$ by the technique introduced in
\cite{Ahl/Bli}: We choose arbitrarily ${d^m}^2$
 Hermitian operators 
$\bar{L}_i\geq 0$, $i=1, \cdots, {d^m}^2$ which span the space of Hermitian operators on
 $H^m$ and denote the trace of $\sum_{i=1}^{{d^m}^2} \bar{L}_i$ by $\lambda$. Now
we  define $L_i:= \frac{1}{\lambda}\bar{L}_i$ for $i\in\{1, \cdots, {d^m}^2\}$ and
$L_{{d^m}^2+1}:= id_{H^m}-\sum_{i=1}^{{d^m}^2} L_i$.

Now we defined the classical arbitrarily varying  channel $\{\grave{W}_{t^m} : t^m \in \theta^m\}$
$: P(\mathbf{U}^m) \rightarrow P(\{1,\cdots , {d^m}^2+3\})$ by 
\[\grave{W}_{t^m}(i\mid p^m):= \begin{cases}\frac{1}{2}\mathrm{tr}\left(\check{W}_{t^m}(p^m)D_i^m\right) ~\text{ for } i\in \{1,2\}\text{ ;}\\
\frac{1}{2}\mathrm{tr}\left(\check{W}_{t^m}(p^m)L_{i-2}\right)~\text{ for } i= 3,\cdots , {d^m}^2+3 \text{ .}\end{cases}\]

Since $\frac{1}{2}\sum_{j = 1}^{2} D_j^m + \frac{1}{2}\sum_{i=1}^{{d^m}^2+1}L_i = id_{H^m}$ 
we have
\[ \sum_{i=1}^{{d^m}^2+3}  \grave{W}_{t^m}(i\mid p^m)=1\] for all $p^m\in P(\mathbf{U}^m)$. Thus
this definition is valid.\vspace{0.2cm}

When $\{\grave{W}_{t^m} : t^m \in \theta^m\}$
is symmetrizable then
there is a  $\{\tau(\cdot\mid a^m):
 a^m\in \mathbf{U}^m\}$ on $\theta^m$ such that 
\[\sum_{t^m\in\theta^m}\tau(t^m\mid a^m)\grave{W}_{t^m}({a'}^m)
=\sum_{t^m\in\theta^m}\tau(t^m\mid {a'}^m)\grave{W}_{t^m}(a^m)\]
for all $i\in\{1,\cdots, {d^m}^2+3\}$. 
This implies that
\[\frac{1}{2}\sum_{t^m\in\theta^m}\tau(t^m\mid a^m)\mathrm{tr}\left(\check{W}_{t^m}({a'}^m)L_i\right)
=\frac{1}{2}\sum_{t^m\in\theta^m}\tau(t^m\mid {a'}^m)\mathrm{tr}\left(\check{W}_{t^m}(a^m)L_i\right)\]
for all $i\in\{1,\cdots, {d^m}^2+1\}$. 

Since 
$\{L_i: i=1, \cdots, {d^m}^2\}$ span  the space of Hermitian operators on $H^m$  
we have 
\[\sum_{t^m\in\theta^m}\tau(t^m\mid a^m)\check{W}_{t^m}({a'}^m)=\sum_{t^m\in\theta^m}\tau(t^m\mid {a'}^m)\check{W}_{t^m}(a^m)\text{ .}\] 
This is a contradiction to our assumption
that $\{\check{W}_{t^m} : t^m \in \theta^m\}$ is not symmetrizable,
 therefore 
$\{\grave{W}_{t^m} : t^m \in \theta^m\}$
is not symmetrizable.\vspace{0.2cm}

\it iii) The
Deterministic Secrecy Capacity
of 
$\{(\grave{W}_{t^m}, \check{V}_{t^m}): t^m \in \theta^m\}$
Is Positive\rm\vspace{0.2cm}

By (\ref{1f1jmspi}) 
for all  $q\in P(\theta)$ and $j\in\{1,2\}$
we have
\begin{equation}
\grave{W}_{q} (j\mid j) \geq \frac{1}{2} - \frac{1}{2}2^{-m^{1/16}\lambda}\text{ ,}
\label{1f1jmspi1} \end{equation}
and for all   $q\in P(\theta)$  and $j\not= i \in\{1 ,2\}$
\begin{equation}
\grave{W}_{q} (j\mid i) \leq \frac{1}{2}2^{-m^{1/16}\lambda}\text{ .}
\label{1f1jmspi2}\end{equation}

We denote the uniform distribution on 
$\{1, 2\}$ by $R'$. 
For  any positive $\zeta'$ if $m$ is 
sufficiently large
by (\ref{1f1jmspi1}) and (\ref{1f1jmspi2}) for all  $q\in P(\theta)$
we have
\begin{align}&
\min_{q \in P({\theta}^m)} I \left(E^m(\cdot\mid R'), \grave{B}_{q}\right)\notag\\
&> 
(\frac{1}{2}-\frac{1}{2}2^{-m^{1/16}\lambda})\log(\frac{1}{2}-\frac{1}{2}2^{-m^{1/16}\lambda})
-(\frac{1}{2}+\frac{1}{2}2^{-m^{1/16}\lambda})\log(\frac{1}{4}+\frac{1}{4}2^{-m^{1/16}\lambda})-\zeta'\notag\\
&\geq \frac{1}{2}-2\zeta'\text{ ,}\label{mqiptmilem}
\end{align} where $\grave{B}_{q}$ is the 
resulting distribution at the output of $\grave{W}_{q}$.\vspace{0.2cm}


Applying the Lemma \ref{eq_9}
and
 (\ref{balf1lnsl1}) if $m$ is 
sufficiently large  for any $n'\in\mathbb{N}$,
 positive $\zeta'$, and
 for all $t^{mn'}$ $=(t^{m}_1,\cdots, t^{m}_{n'})$
$=(t_1,\cdots, t_{m},t_{m+1},\cdots,  t_{2m},t_{2m+1},\cdots, t_{mn'})$  $\in\theta^{mn'}$
we have
\begin{align}\allowdisplaybreaks[2]&\frac{1}{n'}
\max_{ t^{mn'} \in \theta^{mn'}}\chi\left({R'}^{\otimes n'}, \check{Z}_{t^{mn'}}\right)\notag\\
&= \frac{1}{n'} \biggl( S\left(   \frac{1}{2^{n'}} \sum_{j\in \{1,2\}^{n'}} 
\check{V}_{t^{mn'}}(({E^{m}})^{\otimes n'}(\cdot\mid j))\right) 
-     \frac{1}{2^{n'}}\sum_{j\in \{1,2\}^{n'}} S\left(  \check{V}_{t^{mn'}}(({E^{m}})^{\otimes n'}(\cdot\mid j))\right)\biggr)\notag\\
&\leq \frac{1}{n'} \left\vert S\left( \frac{1}{2^{n'}} \sum_{j\in \{1,2\}^{n'}} 
 \check{V}_{t^{mn'}}(({E^{m}})^{\otimes n'}(\cdot\mid j))\right) 
- S\left(  \Theta_{t^{mn'}}\right)  \right\vert\notag\\
&+ \frac{1}{n'} \left\vert  S\left(   \Theta_{t^{mn'}} \right) -
 \frac{1}{2^{n'}} \sum_{j\in \{1,2\}^{n'}} S\left(  \check{V}_{t^{mn'}}(({E^{m}})^{\otimes n'}(\cdot\mid j))\right) \right\vert\notag\\
&= \frac{1}{n'} \left\vert \sum_{i=1}^{n'} \biggl(
 S\left(   \frac{1}{2} \sum_{j\in \{1,2\}}  \check{V}_{t^{m}_{i}}(({E^{m}})(\cdot\mid j))\right) 
- S\left(     \Theta_{t^{m}_{i}}\right)   \biggr) \right\vert\notag\\
&+  \frac{1}{n'}  \left\vert \sum_{i=1}^{n'}  \left( S\left(    \Theta_{t^{m}_{i}} \right) -
\frac{1}{2} \sum_{j\in \{1,2\}}  S\left(  \check{V}_{t^{m}_{i}}(({E^{m}})(\cdot\mid j))\right) \right)\right\vert\notag\\
&\leq 2\cdot 2^{-\sqrt{m}\zeta}\log (d^m-1)+2\cdot h(2^{-\sqrt{m}\zeta})\notag\\
&\leq \zeta'\text{ ,}\label{mtmnitmnc}
\end{align}
where
$\check{Z}_{t^{mn'}}$ is the resulting quantum state at the output of $\check{V}_{t^{mn'}}$.

We choose  $\zeta'<  \frac{1}{18}$ and
a sufficiently large $m$ such that
(\ref{mqiptmilem}) and (\ref{mtmnitmnc}) hold.
Sine $\{\grave{W}_{t^m} : t^m \in \theta^m\}$
is not symmetrizable, by Theorem
\ref{lgwtvttitbaca} the
deterministic secrecy capacity
of 
$\{(\grave{W}_{t^m}, \check{V}_{t^m}): t^m \in \theta^m\}$
is  equal to 
\[ \limsup_{n'\rightarrow \infty} \frac{1}{n'}\max_{p\in P(\mathbf{U}^m)}\min_{q \in P({\theta}^m)}  I(p,\grave{B}_q^{n'}) 
-\max_{t^{mn'}\in \theta^{mn'}}\chi (p,\check{Z}_{t^{mn'}})
 \geq \frac{1}{2}- 3\zeta'> \frac{1}{3}\text{ .}\]\vspace{0.2cm}

\it iv) The Secure Transmission of  the Message with a Deterministic Code\rm\vspace{0.2cm}

Since $(\log n)^2 \gg 3 \log (n^3)$
we can build a $((\log n)^2,n^3)$
code $(\tilde{E}^{(\log n)^3},\{\tilde{S}_i^{(\log n)^3}:i\in\{1,\cdots,n^3\}\})$
such that
\begin{align}& \label{b3mstnomn}1- \min_{t^{(\log n)^3} \in \theta^{(\log n)^3}} \min_{i\in\{1,\cdots,n^3\}} 
\grave{W}_{t^{(\log n)^3}}\left(\tilde{S}_i^{(\log n)^3}\mid \tilde{E}^{(\log n)^3} \left(\cdot \mid i\right)\right)\notag\\
\leq  \varepsilon\text{ ,}\end{align}
and

\begin{equation}
\max_{ t^{(\log n)^3} \in \theta^{(\log n)^3}}\left\|\check{V}_{t^{(\log n)^3}}\left( \tilde{E}^{(\log n)^3} \left(\cdot \mid i\right)\right)- \Theta_{t^{(\log n)^3}}\right\|
\leq \varepsilon\label{mt6l2it}
\end{equation}

for a $\Theta_{t^{(\log n)^3}}\in\mathcal{S}(H^{(\log n)^3})$ which is independent of $j$.

We define $D_j:= L_{j-2}$ for 
$j\in\{3,\cdots,d^{m^2}+3\}$.
For $i\in\{1,\cdots,n^3\}$ we define
\[\tilde{D}_i^{(\log n)^3}:= \frac{1}{2}\sum_{j^m\in \tilde{S}_i^{(\log n)^3}}  D_{j^m} \text{ .}\]
Here for $j^m=(j_1,\cdots,j_m)$ we set $D_{j^m} =D_{j_1}\otimes \cdots \otimes D_{j_m}$.
Since $\sum_{i=1}^{n^3} \tilde{D}_i^{(\log n)^3}$ $=\frac{1}{2}\sum_{j^m\in \{1,\cdots,d^{m^2}+3\}^m} D_{j^m}$ $= id_{H^m}$, 
$\{\tilde{D}_i^{(\log n)^3} : i\in\{1,\cdots,n^3\}\}$ is a valid set of decoding operators.\vspace{0.2cm}

 $(\tilde{E}^{(\log n)^3},\{\tilde{D}_i^{(\log n)^3}:i\in\{1,\cdots,n^3\}\})$
is a $((\log n)^3,n^3)$
code which fulfills
\begin{equation}\min_{ t^{(\log n)^3} \in \theta^{(\log n)^3}}
\frac{1}{n^3}\sum_{i=1}^{n^3}  
\mathrm{tr}(\check{W}_{t^{(\log n)^3}}(\tilde{E}^{(\log n)^3}(\cdot|i))\tilde{D}_i^{(\log n)^3})\geq 1-2^{-n^{1/16}\lambda} 
	\text{ .}\label{f1n3mtwtnr2}\end{equation}\vspace{0.2cm}

\it v) The Secure Transmission of Both the Message and the Randomization Index \rm\vspace{0.2cm}

We choose an arbitrary  positive $\delta$. Let
\[J_n = \frac{1}{n} \left(\max_{U\rightarrow A \rightarrow \{B_q^{n},Z_{t^n}:q,t_n\}} \inf_{B_q \in Conv((B_t)_{t\in \theta})} 
\chi(p_U;B_q^{\otimes n}) - \max_{t^n\in \theta^n}\chi(p_U;Z_{t^n})\right) -\delta\text{ .}\]
By the results of 
\cite{Wi/No/Bo2} if $n$ is sufficiently large there is a $(n, J_n)$
common randomness assisted quantum code
$\left\{\left(\pi \circ E^n, \{P_{\pi} D_j^n P_{\pi}^{T}: j\in\{1,\cdots,J_n\}\}\right): \pi\in\mathsf{S}_n\right\}$,
 a quantum state $\Theta_{t^n}\in\mathcal{S}(H^n)$
such that for  all $t^n\in\theta^n$ 
\begin{align}&
\frac{1}{n!}\frac{1}{J_n}\sum_{\pi\in\mathsf{S}_n} \sum_{j=1}^{J_n}  
\mathrm{tr}\left(\check{W}_{t^n}\left(\pi(E^n(\cdot|j))\right)P_{\pi} D_{\pi(t^n)} P_{\pi}^{T}\right)
\geq 1-2^{-n^{1/16}\lambda} \text{ ,}\end{align}
and for all $t^n\in\theta^n$, $j\in\{1,\cdots,J_n\}$ and all $\pi\in\mathsf{S}_n$ 
\begin{align}&\lVert
 \check{V}_{t^n}\left(\pi(E^n(\cdot|j))\right) -
 P_{\pi}
\Theta_{\pi(t^n)} P_{\pi}^{T} \rVert
< 2^{-\sqrt{n}\zeta}
\label{balf1lnsl12}\end{align}
for a  $\Theta_{t^n}\in\mathcal{S}(H^n)$ which is independent of $j$.

Using  technique in  \cite{Bo/Ca/De} to reduce the amount
of common
randomness if $n$ is sufficiently large
we can find  a set $\{\pi_1,\cdots,\pi_{n^3}\} \subset \mathsf{S}_n$ 
such that
\begin{equation}\max_{t^n\in\theta^n} 
\frac{1}{n^3}\frac{1}{J_n}\sum_{i=1}^{n^3} \sum_{j=1}^{J_n}
\mathrm{tr}\left(\check{W}_{t^n}\left(\pi(E^n(\cdot|j))\right)P_{\pi_i}D_{\pi(t^n)} P_{\pi_i}^{T}\right)\geq 1-2\cdot 2^{-n^{1/16}\lambda} 
	\text{ ,}\label{mtnitn1n3f1jn}\end{equation}
	and
	\begin{align}&\lVert
\check{V}_{t^n}\left(\pi_i(E^n(\cdot|j))\right) -
 P_{\pi_i}
\Theta_{\pi_i(t^n)} P_{\pi_i}^{T} \rVert
<  2^{-\sqrt{n}\zeta}
\label{balf1lnsl13}\end{align}

	Furthermore by the permutation-invariance of $p'$  we also have
$\Theta_{t^n}$ $= P_{\pi}
\Theta_{\pi(t^n)} P_{\pi}^{T}$ for all $\pi\in\mathsf{S}_n$.\vspace{0.2cm}

Now we can construct a $((\log n)^3+n,n^3J_n)$ code
$\bigl(E^{(\log n)^3+n} , \{D_{i,j}^{(\log n)^3+n} : ~ i=1,\cdots,n^3, ~ j = 1,\cdots J_n\}\bigr)$
by
	\begin{equation}E^{(\log n)^3+n}(a^{(\log n)^3+n}\mid i,j) :=\tilde{E}^{(\log n)^3}(a^{(\log n)^3}|i)\cdot
E^n(\pi_i(a^n)|j)
\text{ ,}\end{equation}
for every $a^{{(\log n)^3}+n}= (a^{(\log n)^3},a^{n})\in \mathbf{U}^{{(\log n)^3}+n}$
and 
	\begin{equation}D_{i,j}^{(\log n)^3+n} :=\tilde{D}_i^{(\log n)^3} \otimes (P_{\pi_i}D_j^n P_{\pi_i}^{T})
\text{ .}\end{equation}\vspace{0.2cm}

By (\ref{f1n3mtwtnr2}) and (\ref{mtnitn1n3f1jn})
for every $t^{(\log n)^3+n}= (t^{(\log n)^3},t^n)\in\theta^{(\log n)^3+n}$ we have
\begin{align}& \frac{1}{n^3}\frac{1}{J_n}\sum_{i=1}^{n^3} \sum_{j=1}^{J_n} 
\mathrm{tr}\left(\check{W}_{t^{(\log n)^3+n}} (E^{(\log n)^3+n}(\cdot\mid i,j)) D_{i,j}^{(\log n)^3+n}\right) \notag\\
&=\frac{1}{n^3}\frac{1}{J_n}\sum_{i=1}^{n^3} \sum_{j=1}^{J_n} 
\mathrm{tr}\Biggl(\left[\check{W}_{t^{(\log n)^3}} (\tilde{E}^{(\log n)^3}(\cdot|i))\otimes
 \left(
\check{W}_{t^n}(\pi_i( E^n(\cdot|j)))\right)\right]\left[ \tilde{D}_i^{(\log n)^3} 
\otimes (P_{\pi_i}D_j^n P_{\pi_i}^{T})\right]\Biggr) \allowdisplaybreaks\notag\\
&=\frac{1}{n^3}\sum_{i=1}^{n^3}  
\mathrm{tr}\Biggl(\left[\check{W}_{t^{(\log n)^3}} (\tilde{E}^{(\log n)^3}(\cdot|i))
\tilde{D}_i^{(\log n)^3}\right] \otimes\left[ \frac{1}{J_n}\sum_{j=1}^{J_n}\left(
\check{W}_{t^n}(\pi_i(E^n(\cdot|j)))\right)
 P_{\pi_i}D_j^n P_{\pi_i}^{T}\right]\Biggr) \allowdisplaybreaks\notag\\
&=\frac{1}{n^3}\sum_{i=1}^{n^3}  \Biggl(
\mathrm{tr}\left(\check{W}_{t^{(\log n)^3}} (\tilde{E}^{(\log n)^3}(\cdot|i))
\tilde{D}_i^{(\log n)^3}\right) \cdot  \mathrm{tr}\left( \frac{1}{J_n}\sum_{j=1}^{J_n}\left(
\check{W}_{t^n}(\pi_i(E^n(\cdot|j)))\right)
 P_{\pi_i}D_j^n P_{\pi_i}^{T}\right) \Biggr)\notag\\
&\geq 1-\frac{1}{n^{1/16}}2^{\lambda}-2\cdot 2^{-n^{1/16}\lambda}\notag\\
&\geq 1-  \varepsilon\label{l21gnbrtippnjd}
\end{align}
for any positive 
$\varepsilon$.

By (\ref{mt6l2it}) and (\ref{balf1lnsl13})
for every $t^{(\log n)^3+n}= (t^{(\log n)^3},t^n)\in\theta^{(\log n)^3+n}$ 
and $i \in \{1,\cdots, n^3\}$ and $j\in \{1,\cdots, J_n\}$ we have

\begin{align}& \lVert
\check{V}_{t^{(\log n)^3+n}} (E^{(\log n)^3+n}(\cdot\mid i,j)) -
\Theta_{t^{(\log n)^3}} \otimes
\Theta_{t^n}  \rVert\notag\\
&= \lVert \check{V}_{t^{(\log n)^3}}\left( \tilde{E}^{(\log n)^3} \left(\cdot \mid i\right)\right) \otimes
 \check{V}_{t^n}\left(\pi(E^n(\cdot|j))\right) -
\Theta_{t^{(\log n)^3}} \otimes
\Theta_{t^n}  \rVert\notag\\
&< \frac{1}{\sqrt{n}}2^{\zeta}+ 2^{-\sqrt{n}\zeta}
\text{ .}\end{align}\vspace{0.2cm}

Let 
 $R_{n^3}$ be the 
uniform distribution on $\{1,\cdots,n^3\}$. We define
a random variable $R_{n^3,uni}$ on the
set $\{1,\cdots,n^3\} \times \{1,\cdots,R_n\}$
 by $R_{n^3,uni}:=R_{n^3} \times R_{uni}$.
Applying Lemma \ref{eq_9}
we obtain

\begin{align}&
\max_{ t^{(\log n)^3+n} \in \theta^{(\log n)^3+n}}\chi\left(R_{n^3,J_n}, Z_{t^{(\log n)^3+n}}\right)\notag\\
&\leq\max_{ t^{(\log n)^3} \in \theta^{(\log n)^3}}\chi\left(R_{n^3}, Z_{t^{(\log n)^3}+n}\right)\notag\\
&~+\frac{1}{n^3}\sum_{i=1}^{n^3}\max_{ t^{n} \in \theta^{n}}\chi\left(R_{uni}, \check{V}_{t^{(\log n)^3}} (\tilde{E}^{(\log n)^3}(\cdot|i))\otimes Z_{t^{n},\pi_i}\right)\allowdisplaybreaks\notag\\
&=  \max_{ t^{(\log n)^3} \in \theta^{(\log n)^3}}\Biggl( S\left(\frac{1}{n^3}
\frac{1}{J_n}\sum_{i=1}^{n^3} \sum_{j=1}^{J_n}\check{V}_{t^{(\log n)^3+n}} (E^{(\log n)^3+n}(\cdot\mid i,j)) \right)\allowdisplaybreaks\notag\\
&~-\frac{1}{n^3}\sum_{i=1}^{n^3} S\left(
\frac{1}{J_n}\sum_{j=1}^{J_n}\check{V}_{t^{(\log n)^3+n}} (E^{(\log n)^3+n}(\cdot\mid i,j)) \right)
\Biggr)\allowdisplaybreaks\notag\\
&~ + \max_{ t^{n} \in \theta^{n}}\frac{1}{n^3}\sum_{i=1}^{n^3}\Biggl( S\left(\frac{1}{n^3}
\frac{1}{J_n}\sum_{i=1}^{n^3} \sum_{j=1}^{J_n}\check{V}_{t^{(\log n)^3}} (\tilde{E}^{(\log n)^3}(\cdot|i))\otimes
 \left(
\check{V}_{t^n}(\pi_i( E^n(\cdot|j)))\right) \right)\allowdisplaybreaks\notag\\
&~-  \frac{1}{J_n}\sum_{j=1}^{J_n}S\left(
\check{V}_{t^{(\log n)^3}} (\tilde{E}^{(\log n)^3}(\cdot|i))\otimes
 \left(
\check{V}_{t^n}(\pi_i( E^n(\cdot|j)))\right) \right)\Biggr)\allowdisplaybreaks\notag\\
&\leq \max_{ t^{(\log n)^3} \in \theta^{(\log n)^3}}\Biggl( \Biggl\vert S\left(\frac{1}{n^3}
\frac{1}{J_n}\sum_{i=1}^{n^3} \sum_{j=1}^{J_n}\check{V}_{t^{(\log n)^3+n}} (E^{(\log n)^3+n}(\cdot\mid i,j)) \right)
- \Theta_{t^{(\log n)^3}} \otimes
\Theta_{t^n} \Biggr\vert\allowdisplaybreaks\notag\\
&~+\Biggl\vert\Theta_{t^{(\log n)^3}} \otimes
\Theta_{t^n}-\frac{1}{n^3}\sum_{i=1}^{n^3} S\left(
\frac{1}{J_n}\sum_{j=1}^{J_n}\check{V}_{t^{(\log n)^3+n}} (E^{(\log n)^3+n}(\cdot\mid i,j)) \right)\Biggr\vert
\Biggr)\allowdisplaybreaks\notag\\
&~ + \max_{ t^{n} \in \theta^{n}}\frac{1}{n^3}\Biggl( \Biggl\vert S\left(
\frac{1}{J_n} \sum_{j=1}^{J_n}\check{V}_{t^{(\log n)^3}} (\tilde{E}^{(\log n)^3}(\cdot|i))\otimes
 \left(
\check{V}_{t^n}(\pi_i( E^n(\cdot|j)))\right) \right)
-  \check{V}_{t^{(\log n)^3}} (\tilde{E}^{(\log n)^3}(\cdot|i))\otimes\Theta_{t^n} \Biggr\vert\allowdisplaybreaks\notag\\
&~+\Biggl\vert  \check{V}_{t^{(\log n)^3}} (\tilde{E}^{(\log n)^3}(\cdot|i))\otimes \Theta_{t^n} -  \frac{1}{J_n}\sum_{j=1}^{J_n}S\left(
\check{V}_{t^{(\log n)^3}} (\tilde{E}^{(\log n)^3}(\cdot|i))\otimes
 \left(
\check{V}_{t^n}(\pi_i( E^n(\cdot|j)))\right) \right)\Biggr\vert\Biggr)\allowdisplaybreaks\notag\\
&\leq(\frac{1}{\sqrt{n}}2^{\zeta}+ 2^{-\sqrt{n}\zeta})\log (d^{(\log n)^3} -1)+
h(\frac{1}{\sqrt{n}}2^{\zeta}+ 2^{-\sqrt{n}\zeta}) \notag\\
&~+2^{-\sqrt{n}\zeta}\log (d^n-1) +h(2^{-\sqrt{n}\zeta})\notag\\
&\leq \varepsilon
\end{align}
for any positive 
$\varepsilon$. Here 
$Z_{i,t^{n}}$ is the resulting quantum state at $\check{V}_{t^n}$
after $i\in\{1,\cdots,n^3\}$ has been sent with $E^{(\log n)^3}$.

For any positive $\delta$, if $n$ is large enough we have $\frac{1}{n}\log J_n -\frac{1}{(\log n)^3 +n}\log J_n
\leq \delta$.
Thus 
the  secrecy
rate of $\{(W_t,{V}_t): t \in \theta\}$ to
transmission of both the message and the randomization index 
is large than
	\[
\frac{1}{n}\max_{U\rightarrow A \rightarrow \{B_q,Z_{t}:q,t\}}
	\Bigl(\inf_{B_q \in Conv((B_t)_{t\in \theta})}\chi(p_U;B_q^{\otimes n})-\max_{t^n\in \theta^n}\chi(p_U;Z_{t^n})\Bigr)-2\delta
	\text{ .}\]
\end{proof}

\section{Further Notice on Code Concepts and Applications}

\subsection{Communication With Resources}\label{cwrioppc}

In our previous papers  \cite{Bo/Ca/De} and \cite{Bo/Ca/De2}
we  determined the randomness  assisted secrecy capacities  of arbitrarily varying classical-quantum
wiretap  channels.
\begin{center}\begin{figure}[H]
\includegraphics[width=0.8\linewidth]{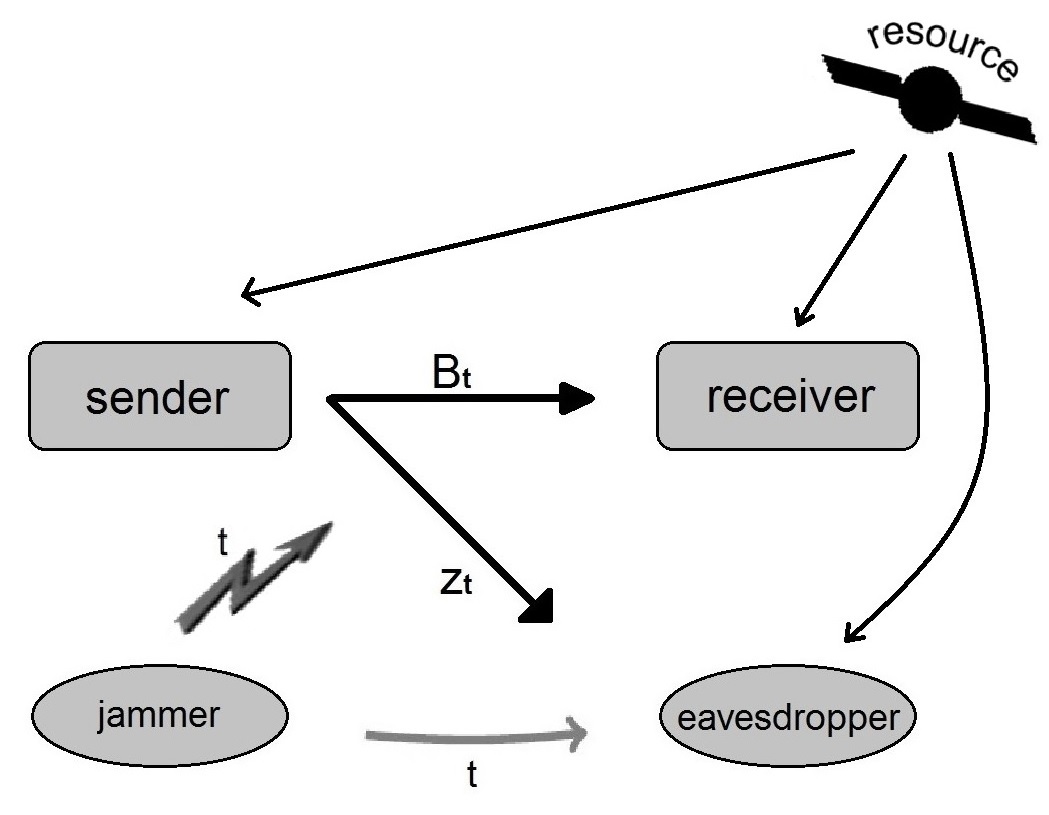}
\caption{
Arbitrarily varying classical-quantum wiretap channel with assistance 
by shared randomness that is not known by the jammer
}
\end{figure}\end{center}
In \cite{Bo/Ca/De} 
we gave an example when the deterministic capacity of an arbitrarily varying classical-quantum
wiretap channel is not equal to its randomness-assisted capacity.
Thus having  resources is very helpful for achieving
a positive secrecy capacity.
For the proofs in  \cite{Bo/Ca/De} and \cite{Bo/Ca/De2} we did not allow
 the jammer to have access 
to the shared randomness.

\begin{center}\begin{figure}[H]
\includegraphics[width=0.8\linewidth]{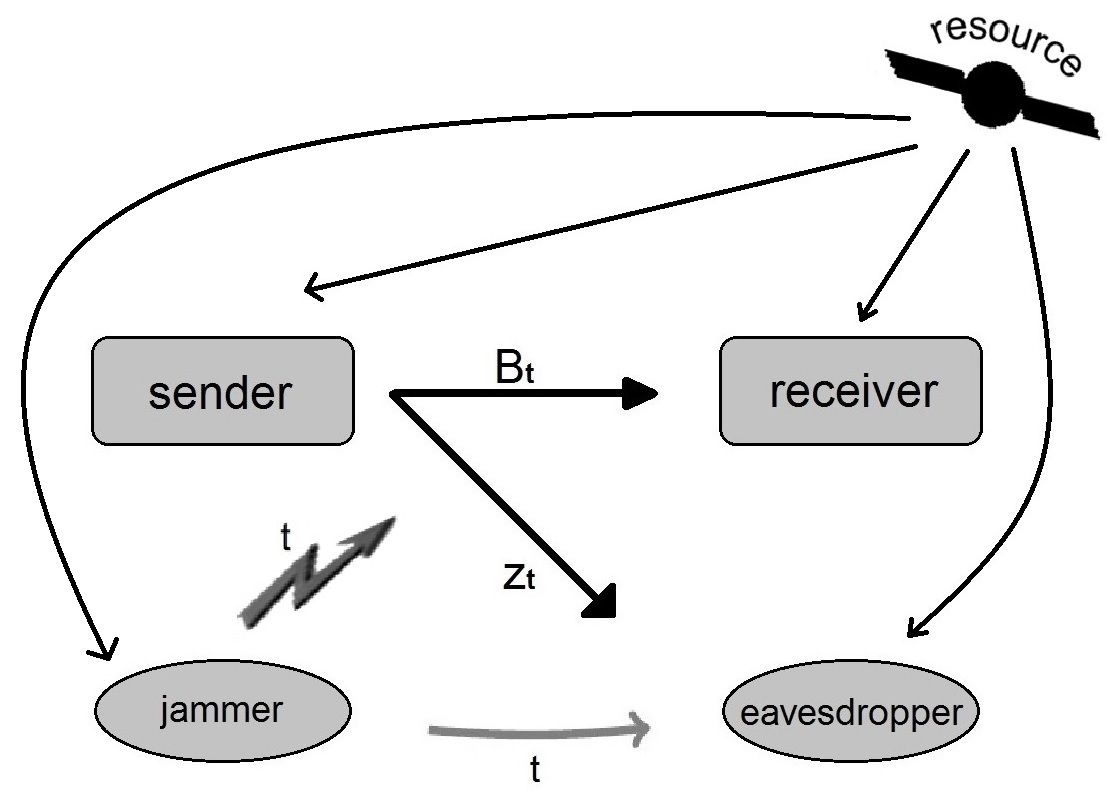}
\caption{
Arbitrarily varying classical-quantum wiretap channel
	with assistance by shared randomness that is known by the jammer}
\end{figure}\end{center}

Now we consider  the case when 
the shared randomness is not secure, i.e.
when the jammer can have access 
to the shared randomness.

\begin{corollary}
	Let $\{({W}_t,{V}_t): t \in \theta\}$ 
	be an arbitrarily varying classical-quantum wiretap channel.
We have
\begin{equation} C_{s}(\{(W_t,{V}_t): t \in \theta\})
=C_{s}(\{(W_t,{V}_t): t \in \theta\};r_{ns})\end{equation}\label{enrtitvwsc}
\end{corollary}

\begin{proof}
Let $\mathcal{C} = \bigl(E, \{D_j^n : j = 1,\cdots J_n\}\bigr)$ be
an  $(n, J_n)$
code such 
\[ \max_{t^n \in \theta^n} P_e(\mathcal{C}, t^n) < \epsilon\text{ ,}\] and
\[\max_{t^n\in\theta^n}
\chi\left(R_{uni};Z_{t^n}\right) < \zeta\text{ .}\]
We define  a $G'$ such that $G'(\{\mathcal{C}\}) = 1$,
it holds
\[ \int_{\Lambda}\max_{t^n\in\theta^n}P_{e}(\mathcal{C}^{\gamma},t^n)d
G'(\gamma) < \epsilon\text{ ,}\]
\[ \int_{\Lambda}\max_{t^n\in\theta^n}
\chi\left(R_{uni},Z_{\mathcal{C}^{\gamma},t^n}\right)d G'(\gamma) < \zeta\text{
.}\] Thus every achievable secrecy rate for
 $\{({W}_t,{V}_t): t \in \theta\}$ is also an
 achievable secrecy rate for
 $\{({W}_t,{V}_t): t \in \theta\}$ under non-secure randomness assisted coding.

Now we assume that there is a $G''$ such that
\[ \int_{\Lambda}\max_{t^n\in\theta^n}P_{e}(\mathcal{C}^{\gamma},t^n)d
G''(\gamma) < \epsilon\text{ ,}\]
\[ \int_{\Lambda}\max_{t^n\in\theta^n}
\chi\left(R_{uni},Z_{\mathcal{C}^{\gamma},t^n}\right)d G''(\gamma) < \zeta\text{
.}\] Then
for any $s^n\in\theta^n$ we have
	\[	\int_{\Lambda}P_{e}(\mathcal{C}^{\gamma},s^n)d
G''(\gamma)
	\leq	\int_{\Lambda}\max_{t^n\in\theta^n}P_{e}(\mathcal{C}^{\gamma},t^n)d
G''(\gamma) < \epsilon\text{ ,}\]
\[ \int_{\Lambda}
\chi\left(R_{uni},Z_{\mathcal{C}^{\gamma},s^n}\right)d G''(\gamma)
\leq \int_{\Lambda}\max_{t^n\in\theta^n}
\chi\left(R_{uni},Z_{\mathcal{C}^{\gamma},t^n}\right)d G''(\gamma) < \zeta\text{ .}\]
Thus every achievable secrecy rate for
 $\{({W}_t,{V}_t): t \in \theta\}$ under non-secure randomness assisted coding is also an
 achievable secrecy rate for
 $\{({W}_t,{V}_t): t \in \theta\}$ under  randomness assisted coding.

Therefore,  \begin{equation} C_{s}(\{(W_t,{V}_t): t \in \theta\})
\leq C_{s}(\{(W_t,{V}_t): t \in \theta\};r_{ns}) \leq
C_{s}(\{(W_t,{V}_t): t \in \theta\};r)\text{ .}
\end{equation}\vspace{0.2cm}

At first let us assume that $\{W_t: t \in \theta\}$
is not symmetrizable.
By \cite{Bo/Ca/De} when $\{W_t: t \in \theta\}$
is not symmetrizable it holds
$C_{s}(\{(W_t,{V}_t): t \in \theta\})$
$= C_{s}(\{(W_t,{V}_t): t \in \theta\};r)$.
Thus when $\{W_t: t \in \theta\}$
is not symmetrizable  we have
\[ C_{s}(\{(W_t,{V}_t): t \in \theta\})
= C_{s}(\{(W_t,{V}_t): t \in \theta\};r_{ns})\text{ .}\]\vspace{0.2cm}

Now let us assume that $\{W_t: t \in \theta\}$
is symmetrizable.
When $\{W_t: t \in \theta\}$
is symmetrizable and  $J_n>1$ hold then by \cite{Bo/Ca/De} for any 
$(n, J_n)$
code $\mathcal{C}$
there is are $t^n \in \theta^n$ and a positive $c$ such that
\[ P_e(\mathcal{C}, t^n) >c\text{ .}\]
Thus when $\{W_t: t \in \theta\}$
is  symmetrizable for any $G$ we have
\[ \int_{\Lambda}\max_{t^n\in\theta^n}P_{e}(\mathcal{C}^{\gamma},t^n)d
G(\gamma) > c\text{ ,}\]
which implies we can only have
$ \int_{\Lambda}\max_{t^n\in\theta^n}P_{e}(\mathcal{C}^{\gamma},t^n)d
G(\gamma)$ $< c$ when $J_n$ is less or equal to $1$. This means
\[C_{s}(\{(W_t,{V}_t): t \in \theta\};r_{ns})= \log 1 = 0\text{ .}\]
By \cite{Bo/Ca/De} when $\{W_t: t \in \theta\}$
is symmetrizable it holds
$C_{s}(\{(W_t,{V}_t): t \in \theta\})= 0$
and therefore  when $\{W_t: t \in \theta\}$
is  symmetrizable we have
\[ C_{s}(\{(W_t,{V}_t): t \in \theta\})
= C_{s}(\{(W_t,{V}_t): t \in \theta\};r_{ns})\text{ .}\]
\end{proof}

In \cite{Bo/Ca/De} we showed that an  arbitrarily varying classical-quantum channel with zero
deterministic secrecy capacity allowed secure transmission if the sender and the
legal receiver had the possibility to use shared randomness
as long as the shared randomness was kept
secret against the jammer.
Corollary \ref{enrtitvwsc} shows that
when the jammer is able have access 
to the outcomes of the shared random experiment
we can only achieve the rate as when
we do not use any shared randomness  at all.
This means
the shared randomness will be
completely useless when it is known by
the jammer.
\vspace{0.2cm}

\begin{center}\begin{figure}[H]
\includegraphics[width=0.8\linewidth]{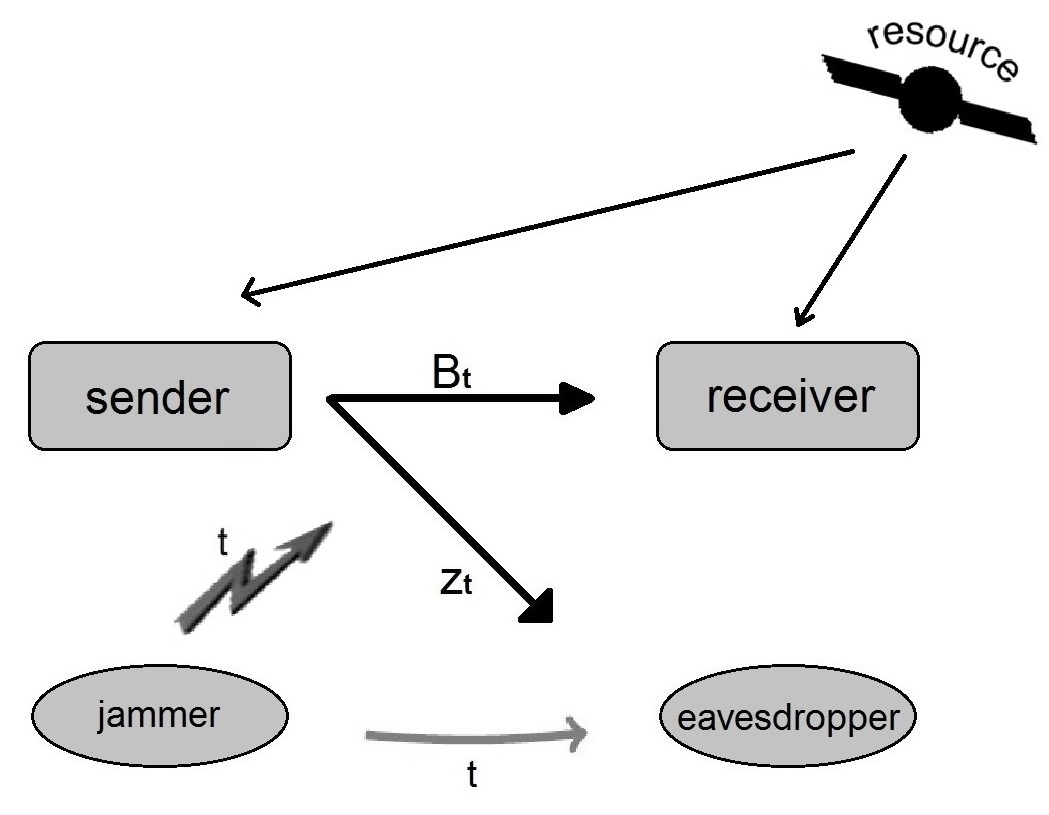}
\caption{
Arbitrarily varying classical-quantum wiretap channel
	with assistance by shared randomness that is not known by the  eavesdropper
}
\end{figure}\end{center}

Applying Theorem \ref{pdwumsfsov} we 
can now determine the
 random   assisted   secrecy capacity 
with the strongest code concept for shared randomness, i.e., 
the randomness which
is secure against both the jammer and
eavesdropping. 

\begin{corollary}
Let   $\theta$ $:=$ $\{1,\cdots,T\}$ be a finite index set.
	Let $\{({W}_t,{V}_t): t \in \theta\}$ 
	be an arbitrarily varying classical-quantum wiretap channel.
	
When $\{{W}_{t} : t\in \theta\}$
is not symmetrizable, we have
	\begin{align}&
C_{key}(\{(W_t,{V}_t): t \in \theta\};g_n)\notag\\
&	=\min\biggl(\lim_{n\rightarrow \infty} \frac{1}{n} \max_{U\rightarrow A \rightarrow \{B_q^{\otimes n},Z_{t^n}:q,t_n\}}
\left(\inf_{B_q \in Conv((B_t)_{t\in \theta})}\chi(p_U;B_q^{\otimes n})-\max_{t^n\in \theta^n}\chi(p_U;Z_{t^n})\right)+g_n,\notag\\
& \max_{U\rightarrow A \rightarrow \{B_q^{\otimes n}:q\}} \inf_{B_q \in Conv((B_t)_{t\in \theta})}\chi(p_U;B_q^{\otimes n})\biggr) \text{ .}
\label{ckwtvttit}\end{align} Here we use the strong code concept.
\label{commpermn2}
\end{corollary}

\begin{remark}
When $g_n$ is positive and independent of $n$, (\ref {ckwtvttit}) always holds and
we do not have to assume that $\{{W}_{t} : t\in \theta\}$
is not symmetrizable.
\end{remark}

\begin{proof}We define 
 $\Gamma_n':= \{1,\cdots, \lceil \frac{n^3}{ \vert\Gamma_n \vert  }\rceil \}$
it holds $n! > \vert \Gamma_n'\times \Gamma_n \vert  \geq n^3$.
Notice that when $g_n$ is positive and independent of $n$ 
we always have have $n! \geq 2^{ng_n} \geq n^3$ for sufficiently large $n$
and thus $\Gamma_n':= \{1 \}$.

We fix a probability distribution
$p\in P(\mathbf{A})$.
Let
\[J_n =\min \left( \lfloor 2^{n\min_{s \in \overline{\theta}}\chi(p;B_s)-\log
L_{n}+ng_n-2n\mu} \rfloor,  \lfloor 2^{n\min_{s \in \overline{\theta}}\chi(p;B_s)-2n\mu}\rfloor\right)\text{ ,}\]
\[L_{n} = \max \left(\lceil 2^{\max_{t^n}\chi(p;Z_{t^n})-ng_n+2n\zeta}\rceil, 1\right)\text{ .}\]
and $p' (x^n):= \begin{cases} \frac{p^{
n}(x^n)}{p^{ n}
(\mathsf{T}^n_{p,\delta})} \text{ ,}& \text{if } x^n \in \mathsf{T}^n_{p,\delta}\text{ ;}\\
0 \text{ ,}& \text{else} \text{ .}\end{cases}$
Let
 $X^{n}$ $:=$ $\{X_{j,l}: j \in \{1, \dots, J_n\}, l \in
\{1, \dots, L_{n}\}\}$ be a family of random variables taking value
according to $p'$.

It holds 
$J_nL_n < 2^{n\min_{s \in \overline{\theta}}\chi(p;B_s)}$
and $L^n2^{g_n } > 2^{\max_{t^n}\chi(p;Z_{t^n})}$.
Similar to the proof of 
Theorem 1 in \cite{Bo/Ca/De} and Theorem 3.1 in   \cite{Bo/Ca/De2}, 
with a positive probability there is a realization $\{x_{j,l}: j , l\}$
of $\{X_{j,l, }: j , l \}$ and a set $\{\pi_{\gamma} :\gamma\in \Gamma_n'\times \Gamma_n  \}$
$\subset \mathsf{S}_n$
with following properties:

There exits a set of decoding operators 
$ \{D_{j,l}:j=1,\cdots,J_n, l=1,\cdots,L_n,\}$
such that  for
every $t^n\in\theta^n$
$\epsilon>0$,  $\zeta>0$, and   sufficiently large $n$,
 \[1- \frac{1}{J_n}\frac{1}{L_n} \frac{1}{\vert \Gamma_n'\times \Gamma_n \vert} \sum_{j=1}^{J_n}  \sum_{j=1}^{L_n} \sum_{\gamma=1}^{\vert \Gamma_n'\times \Gamma_n \vert}  
\mathrm{tr}\left(W_{t^n}(\pi_{\gamma}^{-1}(x_{j,l}))P_{\pi_{\gamma}}^{\dagger}D_{j,l}P_{\pi_{\gamma}}\right)< \epsilon\]
and 
\[ \chi\left(R_{uni},
\frac{1}{J_n}\frac{1}{L_n}\frac{1}{\vert \Gamma_n'\times \Gamma_n \vert}   \sum_{j=1}^{J_n}   \sum_{j=1}^{L_n} \sum_{\gamma=1}^{\vert \Gamma_n'\times \Gamma_n \vert}V_{t^n}(\pi_{\gamma}^{-1}(x_{j,l}))
\right) < \zeta\text{ .}\]

When $\vert\Gamma_n' \vert>1$ holds we use the  strategy of Theorem \ref{pdwumsfsov}
by
build a two-part secure code word, the first part is used  to send $\gamma'\in \Gamma_n'$, the second is used to transmit the message to the legal
receiver.

Thus 
	\begin{align*}&
C_{key}(\{(W_t,{V}_t): t \in \theta\};g_n)\notag\\
&	\geq\min\biggl(\lim_{n\rightarrow \infty} \frac{1}{n} \max_{p }\left(\inf_{B_q \in Conv((B_t)_{t\in \theta})}\chi(p;B_q^{\otimes n})-\max_{t^n\in \theta^n}\chi(p;Z_{t^n})\right)+g_n,\notag\\
& \max_{p }\inf_{B_q \in Conv((B_t)_{t\in \theta})}\chi(p;B_q^{\otimes n})\biggr) \text{ .}\end{align*}

The achievability of
$\lim_{n\rightarrow \infty} \frac{1}{n}\Bigl(\min_{q}\chi(p_U;B_q)$ $-$ $ \max_{t^n}\chi(p_U;Z_{t^n})\Bigr)$ $+g_n$
and $\inf_{B_q \in Conv((B_t)_{t\in \theta})}\chi(p_U;B_q^{\otimes n})$
is then shown via standard arguments.
\vspace{0.2cm}

Now we are going to prove the converse.

	\begin{equation}C_{key}(\{(W_t,{V}_t): t \in \theta\};g_n)
	\leq \max_{U\rightarrow A \rightarrow \{B_q^{\otimes n}:q\}} 
	\inf_{B_q \in Conv((B_t)_{t\in \theta})}\chi(p_U;B_q^{\otimes n})	\label{ckwtvttitgn}\end{equation}
holds trivially.

Let   $(E^{\gamma,(n)},\{D_{j}^{\gamma,(n)}:j\})$ 
be a sequence of $(n,J_n)$ code such
that for every $t^n\in\theta^n$  \[1- \frac{1}{J_n} \frac{1}{2^{ng_n}}\sum_{j=1}^{J_n}     \sum_{\gamma=1}^{2^{ng_n}}
\mathrm{tr}\left(W_{t^n}(E^{\gamma,(n)}(j))D_{j}^{\gamma,(n)}\right)< \epsilon_n\]
and 
\[\chi\left(R_{uni},
\frac{1}{J_n}\frac{1}{2^{ng_n}}  \sum_{j=1}^{J_n}  \sum_{\gamma=1}^{2^{ng_n}}V_{t^n}(E^{\gamma,(n)}(j))
\right) < \zeta_n\text{ 
,}\] 
where $ \lim_{n\to\infty}\epsilon_n=0$ and
$\lim_{n\to\infty}\zeta_n=0$. It is known that for  sufficiently large $n$ we have
	\begin{equation}\log J_n \leq \frac{1}{2^{ng_n}} \sum_{\gamma=1}^{2^{ng_n}}
\chi\left(R_{uni}, B_{q}^{\gamma\otimes n}\right) - 
\chi\left(R_{uni}, Z_{t^n}\right)
\text{ .}	\label{lsntiljnl}\end{equation}

Let $\psi_{q}^{j,\gamma\otimes n}$ $:=  W_{q}^{\otimes n}(E^{\gamma,(n)}(j))$.
We denote $\tilde{B}_{q}^{j\otimes n} := \{W_{q}^{\otimes n}(E^{\gamma,(n)}(j)):\gamma\in \Gamma_n\}$
and $\tilde{B}_{q}^{\otimes n}  := \{\frac{1}{J_n}W_{q}^{\otimes n}(E^{\gamma,(n)}(j)):\gamma\in \Gamma_n\}$.
Let $G_{uni}$ be the uniformly distributed random variable with value in
$\Gamma_n$.

We have
\begin{align}\allowdisplaybreaks[2]&\frac{1}{2^{ng_n}}\sum_{\gamma=1}^{2^{ng_n}}\chi\left(R_{uni};B_{q}^{\gamma\otimes n}\right)
- \chi\left(R_{uni};\frac{1}{2^{ng_n}}\sum_{\gamma=1}^{2^{ng_n}}B_{q}^{\gamma\otimes n}\right)\notag\\
&= \frac{1}{2^{ng_n}}\sum_{\gamma=1}^{2^{ng_n}}S\left(\frac{1}{J_n}\sum_{j=1}^{J_n}\psi_{q}^{j,\gamma\otimes n}\right)
-\frac{1}{2^{ng_n}}\frac{1}{J_n}\sum_{\gamma=1}^{2^{ng_n}}\sum_{j=1}^{J_n}S\left(\psi_{q}^{j,\gamma\otimes n}\right)\notag\\
&-\Biggl[S\left(\frac{1}{2^{ng_n}}\frac{1}{J_n}\sum_{\gamma=1}^{2^{ng_n}}\sum_{j=1}^{J_n}\psi_{q}^{j,\gamma\otimes n}\right)
- \frac{1}{J_n}\sum_{j=1}^{J_n}S\left(\frac{1}{2^{ng_n}}\sum_{\gamma=1}^{2^{ng_n}}\psi_{q}^{j,\gamma\otimes n}\right)\Biggr]\allowdisplaybreaks\notag\\
&= \frac{1}{2^{ng_n}}\sum_{\gamma=1}^{2^{ng_n}}S\left(\frac{1}{J_n}\sum_{j=1}^{J_n}\psi_{q}^{j,\gamma\otimes n}\right)
-S\left(\frac{1}{2^{ng_n}}\frac{1}{J_n}\sum_{\gamma=1}^{2^{ng_n}}\sum_{j=1}^{J_n}\psi_{q}^{j,\gamma\otimes n}\right)\notag\\
&-\Biggl[\frac{1}{2^{ng_n}}\frac{1}{J_n}\sum_{\gamma=1}^{2^{ng_n}}\sum_{j=1}^{J_n}S\left(\psi_{q}^{j,\gamma\otimes n}\right)
- \frac{1}{J_n}\sum_{j=1}^{J_n}S\left(\frac{1}{2^{ng_n}}\sum_{\gamma=1}^{2^{ng_n}}\psi_{q}^{j,\gamma\otimes n}\right)\Biggr]\allowdisplaybreaks\notag\\
&=\frac{1}{J_n}\sum_{j=1}^{J_n}S\left(G_{uni},\tilde{B}_{q}^{j\otimes n}\right)-S\left(G_{uni},\tilde{B}_{q}^{\otimes n} \right)\allowdisplaybreaks\notag\\
&\leq \frac{1}{J_n}\sum_{j=1}^{J_n}S\left(G_{uni},\tilde{B}_{q}^{j\otimes n}\right)\allowdisplaybreaks\notag\\
&\leq \frac{1}{J_n}\sum_{j=1}^{J_n} H\left(G_{uni}\right)\allowdisplaybreaks\notag\\
&= H\left(G_{uni}\right)\notag\\
&=ng_n \text{ .}\label{slgutbqnrn}\end{align} 

By (\ref{ckwtvttitgn}), (\ref{lsntiljnl}), and (\ref{slgutbqnrn}) we have
\begin{align*}&C_{key}(\{(W_t,{V}_t): t \in \theta\};g_n)\\
&\leq \lim_{n\rightarrow \infty} \frac{1}{n} \max_{U\rightarrow A \rightarrow 
\{B_q^{\otimes n},Z_{t^n}:q,t_n\}}\biggl(\inf_{B_q \in Conv((B_t)_{t\in \theta})}\chi(p_U;B_q^{\otimes n})\\
&-\max_{t^n\in \theta^n}\chi(p_U;Z_{t^n})\biggr)+g_n\text{ .}\end{align*}
\end{proof}

\subsection{Some Applications}
\label{SomeApp}

In this section we present some applications
 of our results in \cite{Bo/Ca/De} 
and \cite{Bo/Ca/De2}.

 In \cite{Bo/Ca/De} it has been shown that
the deterministic secrecy capacity of an arbitrarily varying classical-quantum wiretap channel is in general not continuous.
Now we 
deliver the sufficient and
necessary conditions for the continuity 
of 
the capacity function of arbitrarily varying classical-quantum wiretap channels.

\begin{corollary}
For an arbitrarily varying classical-quantum channel
$\{W_t: t \in \theta\}$ we define
\[F(\{W_t: t\}):=\min_{\tau\in C(\theta\mid \mathbf{A}) }\max_{a, a'}\left\|\sum_{t\in\theta}\tau(t\mid a)W_{t}({a'})-\sum_{t\in\theta}\tau(t\mid {a'})W_{t}(a)\right\|_1\text{ ,}\]
where $C(\theta\mid \mathbf{A})$ the set of
parametrized distributions sets $\{\tau(\cdot\mid a):
 a\in \mathbf{A}\}$ on $\theta$.
The statement $F(\{W_t: t\})=0$ is equivalent to $\{W_t: t \in \theta\}$
being symmetrizable.

For an arbitrarily
varying classical-quantum wiretap channel 
$\{(W_t,{V}_t): t \in \theta\}$, where
$W_{t}$  $:$
${P}(\mathbf{A}) \rightarrow \mathcal{S}(H)$ and ${V}_t$ $:$ ${P}(\mathbf{A})
\rightarrow \mathcal{S}(H')$,
 and
a positive $\delta$ let
$\mathsf{C}_{\delta}$ be the set of all
arbitrarily
varying classical-quantum wiretap channels
 $\{({W'}_t,{V'}_t): t \in \theta\}$,
where
${W'}_{t}$  $:$
${P}(\mathbf{A}) \rightarrow \mathcal{S}(H)$ and ${V'}_t$ $:$ $P(\mathbf{A})
\rightarrow \mathcal{S}(H')$,
 such
that
\[\max_{ a\in \mathbf{A}} \|W_t(a)- {W'}_t(a)\|_{1} <  \delta\]
and
\[\max_{ a\in \mathbf{A}} \|V_t(a)- {V'}_t(a)\|_{1} <  \delta\]
for all $t \in \theta$.

 $C_s(\{(W_t,{V}_t): t \})$,
the deterministic secrecy capacity of
arbitrarily
varying classical-quantum wiretap channel
is discontinuous at
$\{({W}_t,{V}_t): t \in \theta\}$ 
if
and only if the following hold:\\[0.15cm]
1)
the secrecy capacity of
$\{({W}_t,{V}_t): t \in \theta\}$  under
common randomness assisted quantum coding
is positive;\\
2) $F(\{W_t: t\})=0$ 
but for every  positive $\delta$ there is a
 $\{({W'}_t,{V'}_t): t \in \theta\}$
$\in$ $\mathsf{C}_{\delta}$ such that
$F(\{{W'}_t: t\})>0$.\label{ftsdmiait}
\end{corollary}

\begin{proof}
At first we assume that the secrecy capacity of
$\{({W}_t,{V}_t): t \in \theta\}$  under
common randomness assisted quantum coding
is positive and $F(\{W_t: t\})=0$.
We choose a positive $\epsilon$
such that $C_s(\{(W_t,{V}_t): t \};cr)$
$-\epsilon$ $:=C$ $>0$.
By Corollary 5.1 in 
\cite{Bo/Ca/De2} 
the secrecy capacity  under
common randomness assisted quantum coding
is continuous. Thus
there exist a 
positive $\delta$ such that the
for all 
$\{({W'}_t,{V'}_t): t \in \theta\}$
$\in$ $\mathsf{C}_{\delta}$ we have
\[ C_s\left(\{({W'}_t,{V'}_t): t \in \theta\};cr\right)\geq 
C_s\left (\{(W_t,{V}_t): t \};cr\right)-\epsilon \text{ .}\]\vspace{0.15cm}

Now we assume that there is a
 $\{({W''}_t,{V''}_t): t \in \theta\}$
$\in$ $\mathsf{C}_{\delta}$ such that
$F(\{{W''}_t: t\})>0$. This means that
$\{{W''}_t: t\}$ is not  symmetrizable.
By Theorem 1 in \cite{Bo/Ca/De} 
it holds
\[ C_s\left(\{({W''}_t,{V''}_t): t \in \theta\}\right)
= C_s\left (\{({W''}_t,{V''}_t): t \};cr\right)\geq C > 0 \text{ .}\]

Since $F(\{W_t: t\})=0$,  $\{{W}_t: t\}$ is symmetrizable.
 By Theorem 1 in \cite{Bo/Ca/De} 
 \[C_s(\{({W}_t,{V}_t): t \in \theta\})=0\text{ .}\] Therefore 
the deterministic secrecy capacity
is discontinuous at
$\{({W}_t,{V}_t): t \in \theta\}$ when 1) and 2) hold.\vspace{0.2cm}

Now let us consider the case when  the deterministic secrecy capacity
is discontinuous at
$\{({W}_t,{V}_t): t \in \theta\}$.\vspace{0.15cm}

We fix a 
$\tau\in C(\theta\mid \mathbf{A})$
and $a$, $a'$ $\in\mathbf{A}$.
The map  
\[\{({W}_t,{V}_t): t \in \theta\}\rightarrow \left\|\sum_{t\in\theta}\tau(t\mid a)W_{t}({a'})-\sum_{t\in\theta}\tau(t\mid {a'})W_{t}(a)\right\|_1\]
is continuous
in the following  sense:
When
$\left\|\sum_{t\in\theta}\tau(t\mid a)W_{t}({a'})-\sum_{t\in\theta}\tau(t\mid {a'})W_{t}(a)\right\|_1$ $=C$
holds then
for every positive $\delta$ and any
$\{({W'}_t,{V'}_t): t \in \theta\}$
$\in$ $\mathsf{C}_{\delta}$ we have
\[\left| \|\sum_{t\in\theta}\tau(t\mid a){W'}_{t}({a'})-\sum_{t\in\theta}\tau(t\mid {a'}){W'}_{t}(a)\|_1 -C\right|\leq 2\delta\text{ .}\]
Thus if for a $\tau\in C(\theta\mid \mathbf{A})$
we have 
$\left\|\sum_{t\in\theta}\tau(t\mid a)W_{t}({a'})-\sum_{t\in\theta}\tau(t\mid {a'})W_{t}(a)\right\|_1$ $=C$ $>0$
for all $a$, $a'$ $\in\mathbf{A}$, 
we also have 
\[\left\|\sum_{t\in\theta}\tau(t\mid a){W'}_{t}({a'})-\sum_{t\in\theta}\tau(t\mid {a'}){W'}_{t}(a)\right\|_1 \geq C-2\delta\text{ .}\]
This means that when  $F(\{W_t: t\})>0$ holds we can find a 
 positive $\delta$ such that
$F(\{{W'}_t: t\})>0$ holds for all
$\{({W'}_t,{V'}_t): t \in \theta\}$
$\in$ $\mathsf{C}_{\delta}$.
By  Theorem 1 in \cite{Bo/Ca/De} it
holds
\[ C_s\left(\{({W'}_t,{V'}_t): t \in \theta\}\right)
= C_s\left (\{({W'}_t,{V'}_t): t \};cr\right)\geq C > 0 \text{ .}\]
By Corollary 5.1 in 
\cite{Bo/Ca/De2} $ C_s\left (\{({W'}_t,{V'}_t): t \};cr\right)$
is continuous. 

Therefore, when  the deterministic secrecy capacity
is discontinuous at
$\{({W}_t,{V}_t): t \in \theta\}$, $F(\{W_t: t\})$ cannot be positive.\vspace{0.15cm}

We consider now that
$F(\{W_t: t\})=0$ holds. By  Theorem 1 in \cite{Bo/Ca/De}
\[C_s(\{({W}_t,{V}_t): t \in \theta\} =0\text{ .}\]
When for every  $\{({W'}_t,{V'}_t): t \in \theta\}$
$\in$ $\mathsf{C}_{\delta}$ we have $F(\{{W'}_t: t\})=0$,
then by  Theorem 1 in \cite{Bo/Ca/De}
 \[C_s(\{({W'}_t,{V'}_t): t \in \theta\})=0\]
and the deterministic secrecy capacity
is thus continuous at
$\{({W}_t,{V}_t): t \in \theta\}$. 

Therefore, when  the deterministic secrecy capacity
is discontinuous at
$\{({W}_t,{V}_t): t \in \theta\}$,  for every  positive $\delta$ there is a
 $\{({W'}_t,{V'}_t): t \in \theta\}$
$\in$ $\mathsf{C}_{\delta}$ such that
$F(\{{W'}_t: t\})>0$.\vspace{0.15cm}

When for every  positive $\delta$ there is a
 $\{({W'}_t,{V'}_t): t \in \theta\}$
$\in$ $\mathsf{C}_{\delta}$ such that
$F(\{{W'}_t: t\})>0$ and $C_s(\{({W}_t,{V}_t): t \in \theta\},cr)$ $=0$
holds,
then by  Theorem 1 in \cite{Bo/Ca/De} we have 
\[C_s(\{({W'}_t,{V'}_t): t \in \theta\})=C_s(\{({W'}_t,{V'}_t): t \in \theta\},cr)\text{ ,}\]
and the deterministic secrecy capacity
is  continuous at
$\{({W}_t,{V}_t): t \in \theta\}$. 

Therefore, when  the deterministic secrecy capacity
is discontinuous at
$\{({W}_t,{V}_t): t \in \theta\}$, $C_s(\{({W}_t,{V}_t): t \in \theta\},cr)$
must be positive.
\end{proof}

\begin{corollary}
Let $\{({W}_t,{V}_t): t \in \theta\}$ be
an  arbitrarily varying classical-quantum  wiretap channel.
When the secrecy capacity of $\{({W}_t,{V}_t): t \in \theta\}$
is positive then there is a  $\delta$
such that for all $\{({W'}_t,{V'}_t): t \in \theta\}$
$\in\mathsf{C}_{\delta}$ we have
\[C_s\left(\{({W'}_t,{V'}_t): t \in \theta\}\right)>0\text{ .}\]
\end{corollary}
\begin{proof}
Suppose we have $C_s(\{({W}_t,{V}_t): t \in \theta\})$ $>0$.
Then $\{W_t: t \in \theta\}$
is not symmetrizable, which means that
$F(\{W_t: t\})$ is positive.
In the proof of Corollary \ref{ftsdmiait} we show that
$F$ is continuous. Thus there is a positive $\delta'$ such that
$F(\{{W'}_t: t\})$ $>0$ for all $\{({W'}_t,{V'}_t): t \in \theta\}$
$\in\mathsf{C}_{\delta'}$ .
When $\{W_t: t \in \theta\}$
is not symmetrizable then we have $C_s(\{({W}_t,{V}_t): t \in \theta\},cr)$ $=$
$C_s(\{({W}_t,{V}_t): t \in \theta\})$ $>0$.
By  Corollary 5.1 in \cite{Bo/Ca/De2}, 
the secrecy capacity  under
common randomness assisted quantum coding is continuous.
Thus there is a positive $\delta''$ such that
 $C_s(\{({W'}_t,{V'}_t): t \in \theta\},cr)$ $>0$
for all $\{({W'}_t,{V'}_t): t \in \theta\}$
$\in\mathsf{C}_{\delta''}$ .
We define $\delta$ $:=$ $\min (\delta',\delta'')$ and the
Corollary is shown.
\end{proof}

One of the properties of classical channels is that in the majority
of cases, if we have a channel system where two sub-channels are
used together, the capacity of this channel system is the sum of the
 two sub-channels' capacities. Particularly,
 a system consisting of two orthogonal classical
channels, where both are ``useless''  in the sense  that they both have zero
capacity for message transmission, the capacity  for message transmission
of the whole system is zero as well (``$0 + 0 = 0$''). 

In contrast to the classical information  theory, it is known that
in quantum information  theory, there are examples of two quantum
channels, $W_1$ and $W_2$,  with zero capacity, which
 allow perfect transmission if they are used together, i.e.,  the   capacity of
their product $W_1\otimes W_2$ is positive. This is due to the fact that there are
different reasons why a quantum channel can have zero capacity. We
call this phenomenon   ``super-activation'' (``$0+0 >0$'').
In \cite{Bo/Ca/De}  super-activation has been shown for  arbitrarily varying classical-quantum
 wiretap channels. Now we  deliver a complete 
characterization of super-activation for  arbitrarily varying classical-quantum
 wiretap channels.

\begin{corollary}
Let $\{({W}_t,{V}_t): t \in \theta\}$ and 
 $\{({W'}_t,{V'}_t): t \in \theta\}$ be
two  arbitrarily varying classical-quantum wiretap  channels.

1) If $C_s(\{({W}_t,{V}_{t}): t\in \theta\})$ $=C_s(\{({W'}_t,{V'}_t): t\in \theta\})$ $=0$
then $C_s(\{W_t\otimes {W'}_{t'},{V}_t\otimes {V'}_{t'}: t , t' \in \theta\})$ is positive
 if and only if $\{ W_t\otimes {W'}_{t'}: t, t' \in \theta \}$ is not symmetrizable and
 $C_s(\{W_t\otimes {W'}_{t'},{V}_t\otimes {V'}_{t'}: t, t'  \in \theta\},cr)$  is positive.

2) If the secrecy capacity   under
common randomness assisted quantum coding shows no super-activation for 
$\{({W}_t,{V}_t): t \in \theta\}$ and 
 $\{({W'}_t,{V'}_t): t \in \theta\}$
 then the secrecy capacity can only then show super-activation for 
$\{({W}_t,{V}_t): t \in \theta\}$ and 
 $\{({W'}_t,{V'}_t): t \in \theta\}$
if one of $\{({W}_t,{V}_t): t \in \theta\}$ and 
 $\{({W'}_t,{V'}_t): t \in \theta\}$  has positive secrecy capacity   under
common randomness assisted quantum coding and a
 symmetrizable legal channel and
while the other one 
 has zero secrecy capacity   under
common randomness assisted quantum 
coding and a non-symmetrizable legal channel.
\end{corollary}

\begin{proof}
By Theorem 1 in \cite{Bo/Ca/De} 
$C_s(\{W_t\otimes {W'}_{t'},{V}_t\otimes {V'}_{t'}: t, t'  \in \theta\})$ 
is equal to
 $C_s(\{W_t\otimes {W'}_{t'},{V}_t\otimes {V'}_{t'}: t, t'  \in \theta\},cr)$ 
when $\{ W_t\otimes {W'}_{t'}: t, t' \in \theta\}$ is not symmetrizable and
to zero when $\{ W_t\otimes {W'}_{t'}: t, t'  \in \theta\}$ is symmetrizable. Thus
1) holds.\vspace{0.2cm}

When  $\{ W_t: t \in \theta\}$ and
 $\{  {W'}_t: t \in \theta\}$ are both
 symmetrizable then there exists two
parametrized set of distributions $\{\tau(\cdot\mid a):
 a\in \mathbf{A}\}$, $\{\tau'(\cdot\mid a):
 a\in \mathbf{A}\}$ 
on $\theta$ such that for all $a$, ${a'}\in \mathbf{A}$,
we have $\sum_{t\in\theta}\tau(t\mid a)W_{t}({a'}) $ $=\sum_{t\in\theta}\tau(t\mid {a'})W_{t}(a)$,
$\sum_{t\in\theta}\tau'(t\mid a){W'}_{t}({a'}) $ $=\sum_{t\in\theta}\tau;(t\mid {a'}){W'}_{t}(a)$,
We can set $\tau((t,t')\mid (a,a'))$ $:=$ $\tau(t\mid a)\tau'(t'\mid a')$
and obtain
\[\sum_{(t,t')\in\theta\times \theta}\tau((t,t')\mid (a_1,a_1'))W_{t}(a_2)\otimes {W'}_{t'}(a_2')
=\sum_{(t,t')\in\theta\times \theta}\tau((t,t')\mid (a_2,a_2'))W_{t}(a_1)\otimes {W'}_{t'}(a_1')\]
for all $(a_1,a_1')$, $(a_2,a_2')$ $\in \mathbf{A} \times \mathbf{A}$,
which means that $\{ W_t\otimes {W'}_t: t, t'  \in \theta\}$ is symmetrizable and
super-activation does not occur because of 1).

When $\{ W_t: t \in \theta\}$ and
 $\{  {W'}_t: t \in \theta\}$ are both
not symmetrizable
then their secrecy capacities  are equal to
their secrecy capacities  under
common randomness assisted quantum 
coding. When  $C_s(\{({W}_t,{V}_t): t \in \theta\},cr)$
 $=C_s(\{({W'}_t,{V'}_t): t \in \theta\},cr)$
$=0$. Because of our assumption
 $C_s(\{W_t\otimes {W'}_{t'},{V}_t\otimes {V'}_{t'}: t, t'  \in \theta\},cr)$ 
$=0$. By 1), super-activation cannot occur.

When one of  $\{ W_t: t \in \theta\}$ and
 $\{  {W'}_t: t \in \theta\}$, say $\{ W_t: t \in \theta\}$
is  not symmetrizable while the other one is symmetrizable,
then $C_s(\{({W}_t,{V}_t): t \in \theta\})$ $=0$ 
indicate that $C_s(\{({W}_t,{V}_t): t  \in \theta\},cr)$ $=0$.
When $C_s(\{({W'}_t,{V'}_t): t \in \theta\},cr)$ is also zero then
by our assumption
super-activation cannot occur. Thus 2) holds.
\end{proof}

\section*{Acknowledgment}
Supports by the Bundesministerium f\"ur Bildung und Forschung (BMBF)
via Grant 16KIS0118K and  16KIS0117K, the German Research Council (DFG) 
via Grant 1129/1-1, the ERC via Advanced Grant IRQUAT, the Spanish MINECO 
via Project
FIS2013-40627-P, 
and the Generalitat de CatalunyaCIRIT via
Project 2014 SGR 966
are gratefully acknowledged.

\end{document}